\documentclass{article}
\usepackage{ijcai25}

\usepackage{times}
\usepackage{soul}
\usepackage{url}
\usepackage[hidelinks]{hyperref}
\usepackage[utf8]{inputenc}
\usepackage[small]{caption}
\usepackage{graphicx}
\usepackage{amsmath}
\usepackage{amsthm}
\usepackage{booktabs}
\usepackage{algorithm}
\usepackage{algorithmic}
\usepackage[switch]{lineno}

\newtheorem{theorem}{Theorem}
\newtheorem{lemma}{Lemma}

\title{Discrete Budget Aggregation: Truthfulness and Proportionality}

\author{
Ulrike Schmidt-Kraepelin$^1$
\and
Warut Suksompong$^2$\and
Markus Utke$^1$
\affiliations
$^1$TU Eindhoven, The Netherlands\\
$^2$National University of Singapore, Singapore\\
}

\usepackage[dvipsnames]{xcolor}
\newif\ifcomments 
\commentstrue

\usepackage{xspace}
\usepackage{natbib} 
\usepackage[capitalize, nameinlink]{cleveref}

\DeclareMathOperator*{\med}{med}
\newcommand{\N}{\mathbb{N}}
\newcommand{\A}{\mathcal{A}}
\newcommand{\F}{\mathcal{F}}
\newcommand{\B}{\mathcal{B}}
\newcommand{\W}{\Phi}
\newcommand{\M}{\mathcal{M}}
\newcommand{\C}{\ensuremath{S^{m}_b}}
\newcommand{\Cc}{\ensuremath{\mathcal{S}_{n,m,b}}}
\newcommand{\Cn}{\ensuremath{S_{\nabla}}}
\newcommand{\D}{\ensuremath{I^{m}_b}}
\newcommand{\Dc}{\ensuremath{\mathcal{I}_{n,m,b}}}
\newcommand{\ph}{\phi}
\newcommand{\ellone}[2]{\lVert#1 - #2\rVert_1}
\newcommand{\Amean}{\mu}
\newcommand{\IM}{\textsc{IndependentMarkets}\xspace}
\newcommand{\UT}{\textsc{Utilitarian}\xspace}
\newcommand{\fIM}{\textsc{FloorIM}\xspace}
\newcommand{\fUT}{\textsc{FloorUtil}\xspace}

\newcommand{\wpo}{\trianglerighteq}

\usepackage{graphicx} 

\usepackage{amsmath}
\usepackage{amsthm}
\usepackage{amssymb}
\usepackage{stmaryrd}

\usepackage{xfrac}

\usepackage{thm-restate}

\usepackage{array}
\usepackage{longtable}
\usepackage{tabularx}
\usepackage{subcaption}
\usepackage{makecell}
\usepackage{enumitem}

\usepackage{tikz}
\usetikzlibrary{
    positioning,
    shapes.geometric,
    shapes.callouts,
    shapes.misc,
    patterns,
    intersections,
    arrows.meta
}
\tikzset{>=latex}

\tikzstyle{aggregate}=[draw, thick, minimum size=8,inner sep=0,regular polygon,regular polygon sides=3]
\tikzstyle{voter}=[draw, thick, minimum size=7,inner sep=0,circle]
\tikzstyle{label}=[text opacity=1]

\definecolor{color1}{HTML}{d97904}
\colorlet{color2}{blue!50!black}
\definecolor{color3}{RGB}{227,114,34}
\definecolor{verylightgray}{HTML}{f1f1f1}

\allowdisplaybreaks

\begin{document}

\maketitle

\begin{abstract}
We study a budget aggregation setting where voters express their preferred allocation of a fixed budget over a set of alternatives, and a mechanism aggregates these preferences into a single output allocation. Motivated by scenarios in which the budget is not perfectly divisible, we depart from the prevailing literature by restricting the mechanism to output allocations that assign integral amounts. This seemingly minor deviation has significant implications for the existence of truthful mechanisms. Specifically, when voters can propose fractional allocations, we demonstrate that the Gibbard--Satterthwaite theorem can be extended to our setting. In contrast, when voters are restricted to integral ballots, we identify a class of truthful mechanisms by adapting \textit{moving-phantom mechanisms} to our context. Moreover, we show that while a weak form of proportionality can be achieved alongside truthfulness, (stronger) proportionality notions derived from approval-based committee voting are incompatible with truthfulness.
\end{abstract}

\section{Introduction}

The summer break is approaching, and you are looking forward to hosting a workshop at your university with participants from around the world.
As the organizer, you need to determine how to allocate the workshop time among paper presentations, poster sessions, and social activities.
Naturally, the participants have varying preferences regarding how the time should be divided.
How should you combine these preferences into the actual allocation?

The problem of aggregating individual preferences on how a budget should be distributed among a set of alternatives is known as \emph{budget aggregation} or \emph{portioning} \citep{freeman2021truthful,caragiannis2022truthful,elkind2023settling,brandt2024optimal,deberg2024truthful,freeman2024project}.
In addition to time, the budget can also represent financial resources, such as when a city council is tasked with allocating its annual funds across different projects.
Several budget aggregation mechanisms have been proposed and investigated in the literature.
An example is the \emph{average mechanism}, which simply returns the average of the preferences of all voters.
Despite its simplicity, this mechanism is susceptible to manipulation: 
if a voter can guess the outcome of the mechanism, she can usually misreport her preference and bring the average closer to her true preference.
In light of this, a number of authors have focused on designing \emph{truthful} mechanisms, i.e., mechanisms for which it is always in the best interest of the voters to report their true preferences.
Notably, \citet{freeman2021truthful} introduced the class of \emph{moving-phantom mechanisms} and demonstrated that every mechanism in this class is truthful.
In addition, a specific moving-phantom mechanism called the \emph{independent markets} mechanism is \emph{(single-minded) proportional}---this means that when every voter is single-minded (i.e., would like the entire budget to be spent on a single alternative), the output of the mechanism coincides with the average of all votes.

As far as we are aware, all prior work on budget aggregation allows a mechanism to output any distribution of the budget.\footnote{\citet{lindner2011zur} considered rules that take integral distributions as their input, but did not place any requirement on the output.} 
However, this can result in ``fractional'' distributions, which may be difficult or even impractical to implement in certain applications.
For instance, a distribution that allots $6.37$ hours from the $10$ available hours at a workshop to paper presentations might be infeasible due to scheduling considerations or the inability to utilize such precise time increments.\footnote{Note that such a distribution can be output, e.g., by the average mechanism, even if all participants submit preferences consisting only of integral numbers of hours.}
Likewise, when allocating funds, it is often more convenient to work with round numbers.
In this paper, we study \emph{discrete budget aggregation}, where an integral budget must be distributed among a set of alternatives in such a way that every alternative receives an integral amount of the budget.
Beyond the allocation of time and money, discrete budget aggregation is generally applicable when the ``budget'' comprises indivisible items, for example, in the distribution of faculty hiring slots among university departments. 

\subsection{Our Contributions}

We study two variants of our model: In the \emph{integral} setting, the voter ballots and the output allocation must be integral, while in the \emph{fractional-input} setting, the voter ballots are allowed to be fractional. For both settings, we establish interesting connections to several social choice frameworks.
\paragraph{Integral Mechanisms: Truthfulness.} 
We explore two approaches for adapting truthful mechanisms from the fractional setting to our integral setting. Firstly, we round the output of fractional mechanisms using \emph{apportionment} methods. 
We show that combining a well-known fractional mechanism with several standard apportionment methods fails truthfulness. Secondly, we translate the idea behind moving-phantom mechanisms directly into our setting. Specifically, we define the class of \textit{integral moving-phantom mechanisms}, and prove that every mechanism in this class is truthful. 
\paragraph{Integral Mechanisms: Proportionality.}
We show that there exist truthful mechanisms (from our class of integral moving-phantom mechanisms) that satisfy \emph{single-minded quota-proportionality}. 
While this property is rather weak, we derive stronger proportionality notions for our setting by viewing it as a subdomain of \emph{approval-based committee elections}.
However, using a computer-aided approach, we show that even the weakest of these notions (called \emph{JR}) is incompatible with truthfulness.
\paragraph{Fractional-Input Mechanisms.} Allowing voters to cast fractional ballots has major implications on the space of truthful mechanisms. Building upon the literature on dictatorial domains, we show that any fractional-input mechanism that is truthful and onto must be dictatorial.
This can be viewed as a variant of the Gibbard--Satterthwaite theorem.

\subsection{Related Work}
\label{sec:related-work}

The analysis of aggregating individual distributions into a collective distribution dates back to the work of \citet{intriligator1973probsc}.
However, Intriligator did not assume that agents possess utility functions and, as a result, did not address the aspect of truthfulness.
Most of the work on truthful budget aggregation thus far assumes that agents are endowed with $\ell_1$ utilities.
Under this assumption, \citet{lindner2008midpoint} and \citet{goel2019knapsack} showed that the mechanism that optimizes utilitarian social welfare (with a certain tie-breaking rule) is truthful.
After \citet{freeman2021truthful} proposed the class of moving-phantom mechanisms, \citet{caragiannis2022truthful} and \citet{freeman2024project} investigated them with respect to the distances of their output from the average distribution, while \citet{deberg2024truthful} presented truthful mechanisms outside this class.
\citet{brandt2024optimal} proved that truthfulness is incompatible with single-minded proportionality and an efficiency notion called \emph{Pareto optimality} under $\ell_1$ utilities, but these properties are compatible under a different utility model.
\citet{elkind2023settling} conducted an axiomatic study of various budget aggregation mechanisms.

Given the integral nature of the output distribution, discrete budget aggregation bears a resemblance to the long-standing problem of apportionment \citep{BaYo82a}.
The main difference is that, in apportionment, the input can be viewed as a single distribution (representing the fractions of voters who support different alternatives) rather than a collection of distributions.
\citet{brill2024approval} studied an approval-based generalization of apportionment, where each voter is allowed to approve multiple alternatives instead of only one.
\citet{DDE+23b} established the incompatibility between truthfulness and representation notions in that setting.

\section{Model and Preliminaries}

For any $z\in\mathbb{N}$, let $[z]$ denote $\{1,\dots,z\}$ and $[z]_0$ denote $\{0,1,\dots,z\}$.
In the setting of budget aggregation, we have a set $[n$] of $n$ \textit{voters} deciding how to distribute a budget of~$b \in \mathbb{N}$ over a set $[m]$ of $m \ge 2$ \textit{alternatives}. 
We write \[\C = \{ v \in [0,b]^m \mid \lVert v \rVert_1 = b \}\] for the set of vectors distributing a budget $b$ over a number of alternatives $m \in \N$, i.e., \C is an $(m-1)$-simplex. Similarly, \[\D = \{ v \in ([b]_0)^m \mid \lVert v \rVert_1 = b \} \subset \C\] denotes the set of vectors \textit{integrally} distributing the budget $b$ over $m$ alternatives. 
We sometimes refer to an element of $\C$ or $\D$ as an \emph{allocation} or a \emph{distribution}.
We denote by $\Cc = (\C)^n$ the set of all \emph{fractional} profiles with $n$ voters, $m$ alternatives, and a budget of $b$, and by $\Dc = (\D)^n$ the set of all integral profiles with the same parameters. 
For each voter~$i$, let $p_i\in \C$ denote her vote, where $p_i = (p_{i,1},\dots,p_{i,m})$.

\paragraph{Budget-Aggregation Mechanisms.} We will consider three types of budget-aggregation mechanisms (or \emph{mechanisms} for short). Generally, a mechanism is a family of functions $\A_{n,m,b}$, one for every triple $n,m,b \in \N$ with $m \ge 2$. We distinguish three types of mechanisms by the type of input and output space of the corresponding functions.

\begin{itemize}
    \item An \textbf{\emph{integral} mechanism} maps any integral profile to an integral aggregate, i.e., $\A_{n,m,b} : \Dc \rightarrow \D$. 
    \item A \textbf{\emph{fractional} mechanism} maps any fractional profile to a fractional aggregate, i.e., $\A_{n,m,b} : \Cc \rightarrow \C$. 
    \item A \textbf{\emph{fractional-input} mechanism} maps any fractional profile to an integral aggregate, i.e., $\A_{n,m,b} : \Cc \rightarrow \D$. 
\end{itemize}
Since $n$, $m$, and $b$ are often clear from context, we slightly abuse notation and write $\mathcal{A}$ instead of $\mathcal{A}_{n,m,b}$. 
While our primary focus is on integral and fractional-input mechanisms, we will build upon fractional mechanisms from the literature. 

We define the \textit{disutility} of voter $i$ with truthful vote $p_i \in \Cc$ towards aggregate $a \in \Cc$ (or $a \in \Dc$) as the $\ell_1$-distance between $p_i$ and $a$, denoted by $\ellone{p_i}{a}$.

\paragraph{Truthfulness.} A mechanism $\A$ is \emph{truthful} if for any $n, m, b \in \N$ with $m \ge 2$ and any profile $P = (p_1, \dots, p_n)$, voter $i \in [n]$, and misreport $p_i^\star$, the following holds for profile $P^\star = (p_1, \dots, p_{i-1}, p_i^\star, p_{i+1}, \dots, p_n)$: \[\ellone{p_i}{\A(P)} \le \ellone{p_i}{\A(P^\star)}.\]
For fractional(-input) mechanisms, both the true profile $P$ and the misreport $P^{\star}$ belong to $\Cc$, while for integral mechanisms these profiles must be from $\Dc$.

\subsection{Moving-Phantom Mechanisms} \label{sec:prelim}

\citet{freeman2021truthful} introduced a class of truthful fractional mechanisms, which we summarize below. 

\paragraph{Moving-Phantom Mechanisms \citep{freeman2021truthful}.} 
For fixed $n, m, b$, a \textit{phantom system} $\F_n$ is a collection of $n+1$ continuous, non-decreasing functions $f_k: [0,1] \rightarrow [0,b]$, with $f_k(0) = 0$ and $f_k(1) \geq b \cdot \frac{n-k}{n}$ for $k \in [n]_0$.
We refer to these functions as \textit{phantom votes} (or just \emph{phantoms}) and to their input as \textit{time}. Any collection of phantom systems $\F = \{\F_n\}_{n \in \N}$ induces a fractional budget aggregation mechanism $\A^\F$, called a \textit{moving-phantom mechanism}. Namely, for a profile $P = (p_1, \dots, p_n) \in \Cc$, an alternative $j \in [m]$, and time $t \in [0,1]$, we denote by $\med(P, \F, j, t) := \med(p_{1,j}, \dots, p_{n,j}, f_0(t), \dots, f_n(t))$ the median of all votes on alternative $j$ and all phantom votes (from $\F_n$) at time $t$. 
Let $t^\star$ be a time such that $\sum_{j \in [m]} \med(P, \F, j, t^\star) = b$; then, the moving-phantom mechanism $\A^\F$ returns the allocation $\A^\F(P) = a$ with $a_j = \med(P, \F, j, t^\star)$ for all $j\in[m]$.
Such $t^\star$ is guaranteed to exist\footnote{We slightly deviate from the definition by \citet{freeman2021truthful} by requiring the sum of medians to reach $b$ instead of $1$. Since we also require phantoms to reach at least $b \cdot \frac{n-k}{n}$ instead of $\frac{n-k}{n}$, this is merely a matter of scaling. \citet[Proposition 3]{freeman2021truthful} showed that requiring $f_k(1)\geq \frac{n-k}{n}$ for all $k \in [n]_0$ implies that the sum of medians at $t=1$ is at least $1$, thus normalization occurs.\label{footnote:normalization}}, and while it may not be unique, the resulting allocation $\A^\F(P)$ is.

We recap two prominent moving-phantom mechanisms from the literature that we will build upon later.

\paragraph{\textsc{IndependentMarkets} \citep{freeman2021truthful}.}
The \textit{\textsc{IndependentMarkets}} mechanism is induced by the phantom system with 
$$f_k(t) = \min(b\cdot(n-k)\cdot t, b)$$  for $k \in [n]_0$ and $n \in \N$. This corresponds to the phantoms moving towards $b$ simultaneously, while being equally spaced (before they get capped at $b$). 

\paragraph{\textsc{Utilitarian} \citep{lindner2008midpoint,goel2019knapsack,freeman2021truthful}.}
The \textit{\textsc{Utilitarian}} mechanism is induced by the phantom system with 
\begin{equation*}
    f_k(t) = 
    \begin{cases}
        0 & \text{ if } t < \frac{k}{n}, \\
        b(tn-k) & \text{ if } \frac{k}{n} \le t \le \frac{k+1}{n}, \\
        b & \text{ if } \frac{k+1}{n} < t
    \end{cases}
\end{equation*}
for $k \in [n]_0$ and $n \in \N$. This corresponds to all phantoms moving towards $b$ one after another (except $f_n$ which stays at 0). 
\UT maximizes utilitarian social welfare (i.e., minimizes the sum of the voters' disutilities).

\section{Integral Mechanisms: Truthfulness} \label{sec:truthful}
We embark on our search for integral mechanisms that are truthful.  
If one of the truthful fractional mechanisms from \Cref{sec:prelim} were guaranteed to output an integral distribution for any integral profile, then this mechanism would directly yield a truthful integral mechanism. 
However, no moving-phantom mechanism satisfies this property---e.g., for the profile $((1,0),(0,1))$, every anonymous and neutral mechanism, and thus every moving-phantom mechanism, must return $(1/2,1/2)$.
In this section, we discuss two approaches for discretizing moving-phantom mechanisms, and exhibit their differing levels of success in achieving truthfulness. 

\subsection{Rounding Fractional Mechanisms}
\label{subsec:rounding_continuous}
Our first approach is to take a fractional mechanism and round its output into an integral output, i.e., we need to map any element of $\C$ to an element of $\D$. 
In fact, this is a well-studied task in the apportionment literature \citep{BaYo82a}; an \emph{apportionment method} is a family of functions (for any $m,b\in \mathbb{N}$) such that $\M_{m,b}: \C \rightarrow \D$. 
Given a fractional mechanism $\A$ and an apportionment method $\M$, we call $\M \circ \A$ the integral mechanism that is \emph{composed of $\A$ and $\M$}. 
Commonly studied apportionment methods include stationary divisor methods, Hamilton's method, and the quota method (see \Cref{app:integral-truthful} for definitions). Stationary divisor methods are parameterized by $\Delta \in [0,1]$, where $\Delta =1$ corresponds to the \textit{Jefferson} (or \textit{d'Hondt}) method and $\Delta = \sfrac{1}{2}$ corresponds to the \emph{Webster} (or \textit{Sainte-Laguë}) method. 
However, applying any of these methods to the outcome of \textsc{IndependentMarkets}
does not yield a truthful mechanism.

\begin{restatable}{theorem}{thmApportionmentNoTie}\label{thm:apportionment-no-tie}
    The composition of \textsc{IndependentMarkets} and the following apportionment methods is not truthful: 
    \begin{itemize}
        \item Hamilton's method
        \item Quota method
        \item Any stationary divisor method for which $\Delta > 0$ and $\frac{2}{\Delta} \not\in \mathbb{N}$
        \item Any stationary divisor method for which $\Delta > 0$ and $\frac{2}{\Delta} \in \mathbb{N}$, if we assume that tie-breaking is in favor of alternatives with higher amounts in the input allocation
    \end{itemize}
\end{restatable}

The proof of \Cref{thm:apportionment-no-tie}, along with all other omitted proofs, can be found in the appendix. 
Clearly, this theorem does not rule out the possibility that combining a different fractional mechanism with an apportionment method gives rise to a truthful integral mechanism; in fact, we will show that this is possible for the \textsc{Utilitarian} mechanism. 
However, the theorem implies that this combination approach does not preserve truthfulness in general.
In the following section, we show that by embedding the rounding within the definition of the moving-phantom mechanism itself, we obtain a general recipe for constructing truthful mechanisms.

\subsection{Integral Moving-Phantom Mechanisms} \label{subsec:discretePhantoms}

The reason why moving-phantom mechanisms can produce non-integral outputs, even when all votes are integral, 
is that the sum of medians can normalize when phantom votes (which are continuous functions) occupy non-integral positions.
We will adjust the phantom systems to the integral setting by modifying them in two ways. 
First, to guarantee integral medians, we let phantom votes increase in discrete steps rather than continuously.
Second, to guarantee normalization, we define phantom votes for each alternative separately; this also reflects the inherent necessity for non-neutrality.

\medskip

For $n, m, b \in \N$, an \textit{integral phantom system} \[\W_{n,m,b} = \{\ph_{k,j} \mid k \in [n]_0,\, j \in [m]\}\] is a set of $(n+1)\cdot m$ non-decreasing functions \[\ph_{k,j}: \mathbb{N} \cup \{0\} \rightarrow [b]_0\] with the following properties, where $z :=  b \cdot m \cdot (n+1)$: 
\begin{enumerate}
    \item $\ph_{k,j}(0) = 0$ and $\ph_{k,j}(z) \geq \lceil \frac{n-k}{n} \cdot b \rceil$ holds for every alternative $j \in [m]$ and every $k \in [n]_0$, and \label{cond1}
    \item  $\sum_{k = 0}^{n} \sum_{j =1}^m \big(\ph_{k,j}(\tau) - \ph_{k,j}(\tau-1)\big) \leq 1$ for all $\tau \in [z]$. \label{cond2}
\end{enumerate}
The idea is that we have $n+1$ phantom votes on each alternative $j \in [m]$, all starting at position $0$ at time $\tau = 0$. In each time step $\tau \rightarrow \tau+1$ at most one of the phantom votes increases its position by $1$, until eventually all phantom votes reach the position $\lceil \frac{n-k}{n} \cdot b \rceil$ or higher. (We will discuss later why this lower bound is useful.)

A family of integral phantom systems $\W = \{\W_{n,m,b} \mid n, m, b \in \N\}$ defines the \emph{integral moving-phantom mechanism} $\A^\W$. 
For a given profile $P = (p_1, \dots, p_n) \in \Dc$, and a time $\tau \in [z]_0$, we are interested in the median of the votes and the phantom votes on each alternative $j$, denoted as 
$$\med(P,\W,j,\tau) = \med(\ph_{0,j}(\tau), \dots, \ph_{n,j}(\tau), p_{1,j}, \dots, p_{n,j}).$$
The integral moving-phantom mechanism $\A^\W$ finds $\tau^{\star}\in [z]_0$ such that $\sum_{j \in [m]} \med(P,\W,j,\tau^{\star}) = b$, and returns $\A^\W(P) =a$ with $a_j = \med(P,\W,j,\tau^{\star})$ for each alternative $j \in [m]$. 
We remark that $\tau^{\star}$ necessarily exists, because by Condition \ref{cond1} of an integral phantom system it holds that $\sum_{j \in [m]} \med(P,\W,j,0) = 0$ and $\sum_{j \in [m]} \med(P,\W,j,z) \geq b$, and by Condition \ref{cond2} it holds that this sum increases by at most $1$ in each time step.\footnote{The statement $\sum_{j \in [m]} \med(P,\W,j,z) \geq b$ follows from the fact that moving-phantom mechanisms are guaranteed to reach normalization when every phantom $k$ reaches $\frac{n-k}{n} \cdot b$ (see \Cref{footnote:normalization}).} While $\tau^{\star}$ is not necessarily unique, the outcome $\A^\W(P)$ is. We illustrate an example in \Cref{fig:discrete_phantom}.

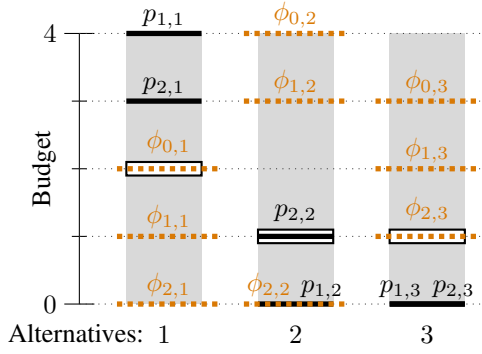
\begin{figure}
    \centering
    
    \begin{tikzpicture}[xscale=1.25, yscale=0.9]
        \def\yaxpos{0.5}
        \def\coorddist{1.4}
        \def\coordwidth{0.4}
        \def\b{4}
        \def\m{3}
        \draw (\yaxpos,0) -- (\yaxpos,\b);  
        \draw (\yaxpos-0.16,0) -- (\yaxpos+0.16,0); 
        \draw (\yaxpos-0.16,\b) -- (\yaxpos+0.16,\b); 
        \foreach \i in {2,...,\b}{
            \draw (\yaxpos-0.08,\i-1) -- (\yaxpos+0.08,\i-1); 
        }
        \node[anchor=east, xshift=-5px] at (\yaxpos,0) {$0$}; 
        \node[anchor=east, xshift=-5px] at (\yaxpos,\b) {$\b$}; 
        \node[anchor=south, rotate=90] at (\yaxpos-0.15,\b/2) {Budget}; 
        \foreach \i in {1,...,\m}{
            \filldraw[fill=black!15,draw=none] (\i*\coorddist-\coordwidth,0) rectangle (\i*\coorddist+\coordwidth,\b); 
        }
        \foreach \i in {0,...,\b}{
            \draw [dotted, darkgray] (\yaxpos,\i) -- (\m*\coorddist+\coordwidth+0.2, \i); 
        }
        \node[anchor=north east, outer sep=4.4pt] at (1*\coorddist,0) {Alternatives:};
        \foreach \i in {1,...,\m}{
            \node[anchor=north, outer sep=5pt] at (\i*\coorddist,0) {$\i$};
        }
        %
        \draw[thick, fill=white] (1*\coorddist-\coordwidth,2-0.1) rectangle (1*\coorddist+\coordwidth,2+0.1); 
        \draw[thick, fill=white] (2*\coorddist-\coordwidth,1-0.1) rectangle (2*\coorddist+\coordwidth,1+0.1); 
        \draw[thick, fill=white] (3*\coorddist-\coordwidth,1-0.1) rectangle (3*\coorddist+\coordwidth,1+0.1);
        \draw[line width=2pt, black] (1*\coorddist-\coordwidth,4) -- (1*\coorddist+\coordwidth,4) node [midway, yshift=5.4pt, xshift=0]{$p_{1,1}$}; 
        \draw[line width=2pt, black] (1*\coorddist-\coordwidth,3) -- (1*\coorddist+\coordwidth,3) node [midway, yshift=5.4pt, xshift=0]{$p_{2,1}$}; 
        \draw[line width=2pt, black] (2*\coorddist-\coordwidth,0) -- (2*\coorddist+\coordwidth,0) node [midway, yshift=5.4pt, xshift=10pt]{$p_{1,2}$}; 
        \draw[line width=2pt, black] (2*\coorddist-\coordwidth,1) -- (2*\coorddist+\coordwidth,1) node [midway, yshift=8.4pt, xshift=0]{$p_{2,2}$}; 
        \draw[line width=2pt, black] (3*\coorddist-\coordwidth,0) -- (3*\coorddist+\coordwidth,0) node [midway, yshift=5.4pt, xshift=-10pt]{$p_{1,3}$}; 
        \draw[line width=2pt, black] (3*\coorddist-\coordwidth,0) -- (3*\coorddist+\coordwidth,0) node [midway, yshift=5.4pt, xshift=10pt]{$p_{2,3}$}; 
        %
        \draw[line width=2pt, color1, dotted] (1*\coorddist-\coordwidth-0.1,0) -- (1*\coorddist+\coordwidth+0.2,0) node [midway, yshift=7pt, xshift=0] {$\ph_{2,1}$}; 
        \draw[line width=2pt, color1, dotted] (1*\coorddist-\coordwidth-0.1,1) -- (1*\coorddist+\coordwidth+0.2,1) node [midway, yshift=7pt, xshift=0] {$\ph_{1,1}$}; 
        \draw[line width=2pt, color1, dotted] (1*\coorddist-\coordwidth-0.1,2) -- (1*\coorddist+\coordwidth+0.2,2) node [midway, yshift=10pt, xshift=0] {$\ph_{0,1}$}; 
        \draw[line width=2pt, color1, dotted] (2*\coorddist-\coordwidth-0.15,0) -- (2*\coorddist+\coordwidth+0.15,0) node [midway, yshift=7pt, xshift=-10pt] {$\ph_{2,2}$}; 
        \draw[line width=2pt, color1, dotted] (2*\coorddist-\coordwidth-0.15,3) -- (2*\coorddist+\coordwidth+0.15,3) node [midway, yshift=7pt, xshift=0] {$\ph_{1,2}$}; 
        \draw[line width=2pt, color1, dotted] (2*\coorddist-\coordwidth-0.15,4) -- (2*\coorddist+\coordwidth+0.15,4) node [midway, yshift=7pt, xshift=0] {$\ph_{0,2}$}; 
        \draw[line width=2pt, color1, dotted] (3*\coorddist-\coordwidth-0.15,1) -- (3*\coorddist+\coordwidth+0.15,1) node [midway, yshift=10pt, xshift=0] {$\ph_{2,3}$};
        \draw[line width=2pt, color1, dotted] (3*\coorddist-\coordwidth-0.15,2) -- (3*\coorddist+\coordwidth+0.15,2) node [midway, yshift=7pt, xshift=0] {$\ph_{1,3}$}; 
        \draw[line width=2pt, color1, dotted] (3*\coorddist-\coordwidth-0.15,3) -- (3*\coorddist+\coordwidth+0.15,3) node [midway, yshift=7pt, xshift=0] {$\ph_{0,3}$}; 
    \end{tikzpicture}
    \caption{Example of an integral moving-phantom mechanism with $n=2$ voters, $m=3$ alternatives, and a budget of $b = 4$.
    The votes on each alternative are marked by (black) solid lines. The phantom positions are shown as (orange) dashed lines. The median vote on each alternative is marked by a rectangle.
    There are two voters with reports $(4, 0, 0)$ and $(3,1,0)$. The figure shows the positions of the phantoms at a time where normalization is reached, i.e., the sum of the median votes is $4$. The returned budget distribution is $(2,1,1)$. 
    }
    \label{fig:discrete_phantom}
\end{figure}

We show in the appendix that any integral phantom system leads to a truthful mechanism. The proof closely follows the proof of truthfulness for fractional moving-phantom mechanisms by \citet{freeman2021truthful}.

\begin{restatable}{theorem}{thmDiscreteTruthful}\label{thm:truthfulness}
    Any integral moving-phantom mechanism is truthful.
\end{restatable}

\paragraph{Rounding Phantom Systems.} \label{subsec:rounding_phantoms}
We can construct integral moving-phantom mechanisms by rounding phantom systems. Let $\F_n = \{f_0(\cdot), \dots, f_{n}(\cdot)\}$ be a phantom system and $\llbracket \cdot \rrbracket$ be a rounding function.\footnote{A rounding function maps any $x \in \mathbb{R}$ to either $\lfloor x \rfloor$ or $\lceil x \rceil$ in such a way that if it maps $x$ to $\lceil x \rceil$, then it also maps every number between $x$ and $\lceil x\rceil$ to $\lceil x\rceil$.} Then, we first track the point in (fractional) time $t \in [0,1]$ at which $\llbracket f_k(t) \rrbracket$ changes for some $k$. We construct an integral phantom system by iterating over these points in time and moving the phantoms $\ph_{k,1}, \dots, \ph_{k,m}$ up by $1$, one after another. We have to be careful when  $\llbracket f_k(t) \rrbracket$ changes for the same $t$ and more than one $k$; in this case, we first move the phantoms with lower $k$. Formally, this leads to the following procedure (see also \Cref{fig:discrete_phantom_from_continuous}): 

\begin{itemize}
    \item Let $0 \leq t_1 < t_2 < \dots < t_{\ell} \leq 1$ be all points in time such that for some $k \in [n]_0$ 
    there is a change in $\llbracket f_k(\cdot) \rrbracket$.
    \item Let $\ph_{k,j}(0) = 0$ for $j \in [m],\, k \in [n]_0$. Let $\tau = 0$. 
    \item For $t_i \in \{t_1, \dots, t_{\ell}\}$, iterate over all $k \in [n]_0$ such that $\llbracket f_k(\cdot) \rrbracket$ changes at $t_i$ and, starting with the lowest such $k$,  do the following for $j \in [m]$ one after another: 
            \begin{itemize}
                \item $\ph_{k,j}(\tau + 1) =\ph_{k,j}(\tau) + 1$,
                \item $\ph_{k',j'}(\tau + 1) =\ph_{k,j}(\tau)$ for all $(j',k') \neq (j,k)$,
                \item increase $\tau$ by $1$.
            \end{itemize}
\end{itemize}

\begin{figure*}
    \centering
    \begin{subfigure}[b]{0.33\textwidth}
        \centering
        \begin{tikzpicture}[scale=0.9]
            \def\yaxpos{0.5}
            \def\coorddist{1.4}
            \def\coordwidth{0.4}
            \def\b{4}
            \def\m{3}
            \draw (\yaxpos,0) -- (\yaxpos,\b);  
            \draw (\yaxpos-0.16,0) -- (\yaxpos+0.16,0); 
            \draw (\yaxpos-0.16,\b) -- (\yaxpos+0.16,\b); 
            \foreach \i in {2,...,\b}{
                \draw (\yaxpos-0.08,\i-1) -- (\yaxpos+0.08,\i-1); 
            }
            \node[anchor=east, xshift=-5px] at (\yaxpos,0) {$0$}; 
            \node[anchor=east, xshift=-5px] at (\yaxpos,\b) {$\b$}; 
            \node[anchor=south, rotate=90] at (\yaxpos-0.15,\b/2) {Budget}; 
            \foreach \i in {1,...,\m}{
                \filldraw[fill=black!15,draw=none] (\i*\coorddist-\coordwidth,0) rectangle (\i*\coorddist+\coordwidth,\b); 
            }
            \foreach \i in {0,...,\b}{
                \draw [dotted, darkgray] (\yaxpos,\i) -- (\m*\coorddist+\coordwidth+0.2, \i); 
            }
            \node[anchor=north east, outer sep=4.4pt] at (1*\coorddist,0) {Alternatives:};
            \foreach \i in {1,...,\m}{
                \node[anchor=north, outer sep=5pt] at (\i*\coorddist,0) {$\i$};
            }
            %
            \node[anchor=south, yshift=5px] at (\m*\coorddist/2+\coorddist/2, \b) {$\tau$}; 
            %
            \draw[line width=2pt, color2] (1*\coorddist-\coordwidth-0.15,0) -- (\m*\coorddist+\coordwidth+0.15,0) node [xshift=7pt] {$f_2$}; 
            \draw[line width=2pt, color2] (1*\coorddist-\coordwidth-0.15,1.4) -- (\m*\coorddist+\coordwidth+0.15,1.4) node [xshift=7pt] {$f_1$}; 
            \draw[line width=2pt, color2] (1*\coorddist-\coordwidth-0.15,2.8) -- (\m*\coorddist+\coordwidth+0.15,2.8) node [xshift=7pt] {$f_0$}; 
            %
            \draw[line width=2pt, color1, dotted] (1*\coorddist-\coordwidth-0.1,0) -- (1*\coorddist+\coordwidth+0.2,0) node [midway, yshift=7pt] {$\ph_{2,1}$}; 
            \draw[line width=2pt, color1, dotted] (1*\coorddist-\coordwidth-0.1,1) -- (1*\coorddist+\coordwidth+0.2,1) node [midway, yshift=7pt] {$\ph_{1,1}$}; 
            \draw[line width=2pt, color1, dotted] (1*\coorddist-\coordwidth-0.1,2) -- (1*\coorddist+\coordwidth+0.2,2) node [midway, yshift=7pt] {$\ph_{0,1}$}; 
            \draw[line width=2pt, color1, dotted] (2*\coorddist-\coordwidth-0.15,0) -- (2*\coorddist+\coordwidth+0.15,0) node [midway, yshift=7pt] {$\ph_{2,2}$}; 
            \draw[line width=2pt, color1, dotted] (2*\coorddist-\coordwidth-0.15,1) -- (2*\coorddist+\coordwidth+0.15,1) node [midway, yshift=7pt] {$\ph_{1,2}$}; 
            \draw[line width=2pt, color1, dotted] (2*\coorddist-\coordwidth-0.15,2) -- (2*\coorddist+\coordwidth+0.15,2) node [midway, yshift=7pt] {$\ph_{0,2}$}; 
            \draw[line width=2pt, color1, dotted] (3*\coorddist-\coordwidth-0.15,0) -- (3*\coorddist+\coordwidth+0.15,0) node [midway, yshift=7pt] {$\ph_{2,3}$};
            \draw[line width=2pt, color1, dotted] (3*\coorddist-\coordwidth-0.15,1) -- (3*\coorddist+\coordwidth+0.15,1) node [midway, yshift=7pt] {$\ph_{1,3}$}; 
            \draw[line width=2pt, color1, dotted] (3*\coorddist-\coordwidth-0.15,2) -- (3*\coorddist+\coordwidth+0.15,2) node [midway, yshift=7pt] {$\ph_{0,3}$}; 
        \end{tikzpicture}
    \end{subfigure}%
    \hfill
    \begin{subfigure}[b]{0.33\textwidth}
        \centering
        \begin{tikzpicture}[scale=0.9]
            \def\yaxpos{0.5}
            \def\coorddist{1.4}
            \def\coordwidth{0.4}
            \def\b{4}
            \def\m{3}
            \draw (\yaxpos,0) -- (\yaxpos,\b);  
            \draw (\yaxpos-0.16,0) -- (\yaxpos+0.16,0); 
            \draw (\yaxpos-0.16,\b) -- (\yaxpos+0.16,\b); 
            \foreach \i in {2,...,\b}{
                \draw (\yaxpos-0.08,\i-1) -- (\yaxpos+0.08,\i-1); 
            }
            \node[anchor=east, xshift=-5px] at (\yaxpos,0) {$0$}; 
            \node[anchor=east, xshift=-5px] at (\yaxpos,\b) {$\b$}; 
            \node[anchor=south, rotate=90] at (\yaxpos-0.15,\b/2) {Budget}; 
            \foreach \i in {1,...,\m}{
                \filldraw[fill=black!15,draw=none] (\i*\coorddist-\coordwidth,0) rectangle (\i*\coorddist+\coordwidth,\b); 
            }
            \foreach \i in {0,...,\b}{
                \draw [dotted, darkgray] (\yaxpos,\i) -- (\m*\coorddist+\coordwidth+0.2, \i); 
            }
            \node[anchor=north east, outer sep=4.4pt] at (1*\coorddist,0) {Alternatives:};
            \foreach \i in {1,...,\m}{
                \node[anchor=north, outer sep=5pt] at (\i*\coorddist,0) {$\i$};
            }
            %
            \node[anchor=south, yshift=5px] at (\m*\coorddist/2+\coorddist/2, \b) {$\tau+1$}; 
            %
            \draw[line width=2pt, color2] (1*\coorddist-\coordwidth-0.15,0) -- (\m*\coorddist+\coordwidth+0.15,0) node [xshift=7pt] {$f_2$}; 
            \draw[line width=2pt, color2] (1*\coorddist-\coordwidth-0.15,1.5) -- (\m*\coorddist+\coordwidth+0.15,1.5) node [xshift=7pt] {$f_1$}; 
            \draw[line width=2pt, color2] (1*\coorddist-\coordwidth-0.15,3.0) -- (\m*\coorddist+\coordwidth+0.15,3.0) node [xshift=7pt] {$f_0$}; 
            %
            \draw[line width=2pt, color1, dotted] (1*\coorddist-\coordwidth-0.1,0) -- (1*\coorddist+\coordwidth+0.2,0) node [midway, yshift=7pt] {$\ph_{2,1}$}; 
            \draw[line width=2pt, color1, dotted] (1*\coorddist-\coordwidth-0.1,1) -- (1*\coorddist+\coordwidth+0.2,1) node [midway, yshift=7pt] {$\ph_{1,1}$}; 
            \draw[line width=2pt, color1, dotted] (1*\coorddist-\coordwidth-0.1,3) -- (1*\coorddist+\coordwidth+0.2,3) node [midway, yshift=7pt] {$\ph_{0,1}$}; 
            \draw[line width=2pt, color1, dotted] (2*\coorddist-\coordwidth-0.15,0) -- (2*\coorddist+\coordwidth+0.15,0) node [midway, yshift=7pt] {$\ph_{2,2}$}; 
            \draw[line width=2pt, color1, dotted] (2*\coorddist-\coordwidth-0.15,1) -- (2*\coorddist+\coordwidth+0.15,1) node [midway, yshift=7pt] {$\ph_{1,2}$}; 
            \draw[line width=2pt, color1, dotted] (2*\coorddist-\coordwidth-0.15,2) -- (2*\coorddist+\coordwidth+0.15,2) node [midway, yshift=7pt] {$\ph_{0,2}$}; 
            \draw[line width=2pt, color1, dotted] (3*\coorddist-\coordwidth-0.15,0) -- (3*\coorddist+\coordwidth+0.15,0) node [midway, yshift=7pt] {$\ph_{2,3}$};
            \draw[line width=2pt, color1, dotted] (3*\coorddist-\coordwidth-0.15,1) -- (3*\coorddist+\coordwidth+0.15,1) node [midway, yshift=7pt] {$\ph_{1,3}$}; 
            \draw[line width=2pt, color1, dotted] (3*\coorddist-\coordwidth-0.15,2) -- (3*\coorddist+\coordwidth+0.15,2) node [midway, yshift=7pt] {$\ph_{0,3}$}; 
        \end{tikzpicture}
    \end{subfigure}%
    \hfill
    \begin{subfigure}[b]{0.33\textwidth}
        \centering
        \begin{tikzpicture}[scale=0.9]
            \def\yaxpos{0.5}
            \def\coorddist{1.4}
            \def\coordwidth{0.4}
            \def\b{4}
            \def\m{3}
            \draw (\yaxpos,0) -- (\yaxpos,\b);  
            \draw (\yaxpos-0.16,0) -- (\yaxpos+0.16,0); 
            \draw (\yaxpos-0.16,\b) -- (\yaxpos+0.16,\b); 
            \foreach \i in {2,...,\b}{
                \draw (\yaxpos-0.08,\i-1) -- (\yaxpos+0.08,\i-1); 
            }
            \node[anchor=east, xshift=-5px] at (\yaxpos,0) {$0$}; 
            \node[anchor=east, xshift=-5px] at (\yaxpos,\b) {$\b$}; 
            \node[anchor=south, rotate=90] at (\yaxpos-0.15,\b/2) {Budget}; 
            \foreach \i in {1,...,\m}{
                \filldraw[fill=black!15,draw=none] (\i*\coorddist-\coordwidth,0) rectangle (\i*\coorddist+\coordwidth,\b); 
            }
            \foreach \i in {0,...,\b}{
                \draw [dotted, darkgray] (\yaxpos,\i) -- (\m*\coorddist+\coordwidth+0.2, \i); 
            }
            \node[anchor=north east, outer sep=4.4pt] at (1*\coorddist,0) {Alternatives:};
            \foreach \i in {1,...,\m}{
                \node[anchor=north, outer sep=5pt] at (\i*\coorddist,0) {$\i$};
            }
            %
            \node[anchor=south, yshift=5px] at (\m*\coorddist/2+\coorddist/2, \b) {$\tau+3$}; 
            %
            \draw[line width=2pt, color2] (1*\coorddist-\coordwidth-0.15,0) -- (\m*\coorddist+\coordwidth+0.15,0) node [xshift=7pt] {$f_2$}; 
            \draw[line width=2pt, color2] (1*\coorddist-\coordwidth-0.15,1.5) -- (\m*\coorddist+\coordwidth+0.15,1.5) node [xshift=7pt] {$f_1$}; 
            \draw[line width=2pt, color2] (1*\coorddist-\coordwidth-0.15,3.0) -- (\m*\coorddist+\coordwidth+0.15,3.0) node [xshift=7pt] {$f_0$}; 
            %
            \draw[line width=2pt, color1, dotted] (1*\coorddist-\coordwidth-0.1,0) -- (1*\coorddist+\coordwidth+0.2,0) node [midway, yshift=7pt] {$\ph_{2,1}$}; 
            \draw[line width=2pt, color1, dotted] (1*\coorddist-\coordwidth-0.1,1) -- (1*\coorddist+\coordwidth+0.2,1) node [midway, yshift=7pt] {$\ph_{1,1}$}; 
            \draw[line width=2pt, color1, dotted] (1*\coorddist-\coordwidth-0.1,3) -- (1*\coorddist+\coordwidth+0.2,3) node [midway, yshift=7pt] {$\ph_{0,1}$}; 
            \draw[line width=2pt, color1, dotted] (2*\coorddist-\coordwidth-0.15,0) -- (2*\coorddist+\coordwidth+0.15,0) node [midway, yshift=7pt] {$\ph_{2,2}$}; 
            \draw[line width=2pt, color1, dotted] (2*\coorddist-\coordwidth-0.15,1) -- (2*\coorddist+\coordwidth+0.15,1) node [midway, yshift=7pt] {$\ph_{1,2}$}; 
            \draw[line width=2pt, color1, dotted] (2*\coorddist-\coordwidth-0.15,3) -- (2*\coorddist+\coordwidth+0.15,3) node [midway, yshift=7pt] {$\ph_{0,2}$}; 
            \draw[line width=2pt, color1, dotted] (3*\coorddist-\coordwidth-0.15,0) -- (3*\coorddist+\coordwidth+0.15,0) node [midway, yshift=7pt] {$\ph_{2,3}$};
            \draw[line width=2pt, color1, dotted] (3*\coorddist-\coordwidth-0.15,1) -- (3*\coorddist+\coordwidth+0.15,1) node [midway, yshift=7pt] {$\ph_{1,3}$}; 
            \draw[line width=2pt, color1, dotted] (3*\coorddist-\coordwidth-0.15,3) -- (3*\coorddist+\coordwidth+0.15,3) node [midway, yshift=7pt] {$\ph_{0,3}$}; 
        \end{tikzpicture}
    \end{subfigure}%
    \caption{Illustration showing how to construct the integral phantom system $\W$ from a fractional phantom system $\F$. In this example, we have $n = 2$, $m = 3$, $b = 4$, the fractional phantom system is \textsc{IndependentMarkets}, and rounding is done using the floor function. Each fractional phantom $f_k$ is drawn as a (blue) line spanning all alternatives and each integral phantom $\ph_{k,j}$ is drawn as an (orange) dashed line. In the left figure (discrete time step $\tau$), no fractional phantom is crossing an integer value and all integral phantoms correspond to a rounded fractional phantom. As time progresses, the upper fractional phantom $f_0$ reaches $3$, at which point the corresponding integral phantoms should move from $2$ to $3$. To guarantee a time of normalization, they move one after another, as illustrated in the middle and right figures.
    }
    \label{fig:discrete_phantom_from_continuous}
\end{figure*}
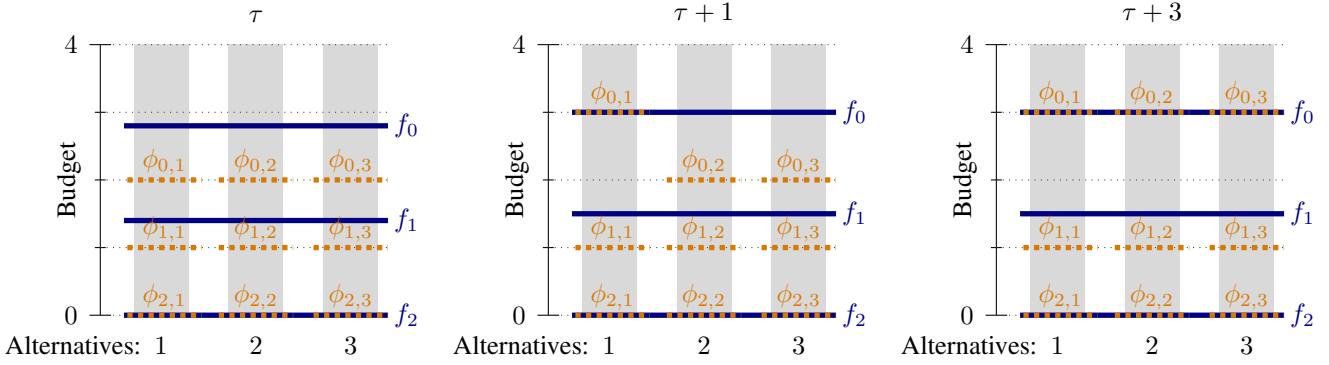

Two integral moving-phantom mechanisms that will be of particular interest are the combination of a variant of \IM and the floor rounding function (referred to as \fIM), and the combination of \UT and the floor rounding function (referred to as \fUT). We show that \fUT is equivalent to the mechanism induced by combining \UT with Hamilton's apportionment method via the process described in \Cref{subsec:rounding_continuous}. In particular, this shows that the approach from \Cref{subsec:rounding_continuous} can lead to truthful mechanisms. 

\begin{restatable}{proposition}{floorUtilEqualsUtilHamilton}\label{prop:eqUtilHamilton}
    The composition of \UT and Hamilton's method (with tie-breaking by indices of alternatives) is equivalent to \fUT.  
\end{restatable}
 
In the following section, we show that \fIM offers a desirable property beyond truthfulness.

\section{Integral Mechanisms: Proportionality} \label{sec:truthful_proportional}

Having established the existence of truthful mechanisms in the integral setting, we next examine how well these mechanisms perform with respect to other properties. We focus on proportionality, i.e., we want a mechanism to reflect the preferences of voter groups proportionally.
There exists a proportionality notion in the fractional setting, which requires a mechanism to output the average distribution if all voters are single-minded.
A voter~$i$ is said to be \textit{single-minded} if $p_{i,j} = b$ for some alternative $j$ (and therefore $p_{i,j'} = 0$ for all alternatives $j'\ne j$). We call a profile single-minded if all voters are single-minded, and define the average allocation $\mu(P)$ where $\mu(P)_j = \frac{1}{n}\sum_{i \in N} p_{i,j}$ for each $j \in [m]$.

\paragraph[Single-Minded~Proportionality]{Single-Minded Proportionality\footnote{\citet{freeman2021truthful} called this axiom \textit{proportionality}; we deviate from this to distinguish it from other proportionality notions.} \citep{freeman2021truthful}.} A fractional budget-aggregation mechanism $\A$ is \emph{single-minded proportional} if for any $n, m, b \in \N$ with $m \ge 2$ and any single-minded profile $P$, it holds that $\A(P) = \mu(P)$.

Clearly, outputting exactly the average is not always possible in the integral setting. We therefore adapt the axiom to make it satisfiable in our setting. 

\paragraph{Single-Minded Quota-Proportionality.} An integral budget-aggregation mechanism $\A$ is \emph{single-minded quota-proportional} if for any $n, m, b \in \N$ with $m \ge 2$ and any single-minded profile $P$, the output allocation $a = \A(P)$ satisfies $a_j \in \{\lfloor \Amean(P)_j \rfloor, \lceil \Amean(P)_j \rceil\}$ for all $j \in [m]$.

We establish the existence of truthful, single-minded quota-proportional mechanisms by adapting the fractional phantom system of single-minded proportional moving-phantom mechanisms and then translating them into integral mechanisms as described in \Cref{subsec:discretePhantoms}. For $n,b \in \N$, we call a (fractional) phantom system $\F_n = \{f_0, \dots, f_n\}$ \textit{upper-quota capped} if for all $k \in [n]_0$ we have $f_k(1) = \lceil b \cdot \frac{n-k}{n} \rceil$. 

\begin{restatable}{theorem}{thmQuotaProportionality}\label{thm:quota_proportionality}
    For any single-minded proportional and upper-quota capped phantom system $\F$, the integral moving-phantom mechanism induced by $\F$ and the floor function satisfies single-minded quota-proportionality.
\end{restatable}

We can transform any phantom system $\F_n$ into an upper-quota capped system $\F'_n$: First extend $\F_n$ to guarantee $f_k(t) \geq \lceil b \cdot \frac{n-k}{n} \rceil$ (if necessary), then set $f_k'(t) = \min(f_k(t), \lceil b \cdot \frac{n-k}{n} \rceil)$. Generally, $\A^{\F}$ and $\A^{\F'}$ need not be equivalent, but in the case of the \IM phantom system---call it $\mathcal{G}$---they are. 
We define \fIM as the integral moving-phantom mechanism induced by $\mathcal{G}'$ and the floor function.  
\Cref{thm:quota_proportionality} then implies that \fIM is single-minded quota-proportional.
We remark that the theorem does not hold if we use $\mathcal{G}'$ (or $\mathcal{G}$) and the ceiling function. For example, consider the instance with $n = 6$, $m = 4$, and $b = 4$, where three voters vote $(4, 0, 0, 0)$ and one voter each votes $(0, 4, 0, 0)$, $(0, 0, 4, 0)$, and $(0, 0, 0, 4)$. 
The upper $n$ phantoms are immediately rounded to~$1$, leading to the output $(1,1,1,1)$, which violates single-minded quota-proportionality for the first alternative.

Single-minded quota-proportionality is a rather weak proportionality notion, as it only applies to a highly restricted subclass of profiles.
Consider, for example, the non-single-minded profile $P = (p_1, p_2)$ for $n = 2$, $m = 4$, and $b = 2$ with $p_1 = (1,1,0,0)$ and $p_2 = (0,0,1,1)$. Clearly, a desirable outcome should allocate $1$ to either alternative $1$ or $2$ and also $1$ to either alternative $3$ or $4$, so that both voters are equally represented. However, integral moving-phantom mechanisms do not consider which of the votes on different alternatives come from the same voter, and may therefore (depending on the tie-breaking) return an allocation like $(1,1,0,0)$.

In order to define notions that capture a wider range of scenarios, we demonstrate that our setting can be interpreted as a subdomain of the well-studied domain of \emph{approval-based committee voting} \citep{lackner2023multi}. 
This allows us to import established axioms of proportional representation to our setting. 
We show that the failure to satisfy these axioms is not a weakness of integral moving-phantom mechanisms per se, but rather stems from more general limitations of truthful mechanisms.

\paragraph{Connection to Approval-Based Committee Voting.} In approval-based committee voting, we have a set of voters $N$, a set of candidates $M$, and a committee size $k \in \N$. Each voter $i$ approves a subset of the candidates $A_i \subseteq M$, and a \textit{voting rule} chooses a \textit{committee} $W \subseteq M$ of size $|W| = k$. The satisfaction of a voter $i$ with a committee $W$ is $|A_i \cap W|$.

We can interpret any instance of our setting as an approval-based committee election with an equivalent utility model (see also \citet{goel2019knapsack}).
Let $P = (p_1, \dots, p_n)$ be a profile in the integral budget aggregation setting. We set $M = \{c_{j,\ell} \mid j \in [m],\, \ell \in [b]\}$ to be the set of candidates, $k = b$, and $A_i = \bigcup_{j \in [m]} \{c_{j,\ell} \mid \ell \in [p_{i,j}]\}$. Intuitively, for each alternative we create $b$ (ordered) candidates corresponding to it, and a voter approves as many of these candidates (in order) as the amount of budget that she would like to allocate to that alternative. This translation is illustrated in \Cref{fig:approval_voting}. Any chosen allocation $a \in \D$ can similarly be translated into a committee $W = \bigcup_{j \in [m]} \{c_{j,\ell} \mid \ell \in [a_j]\}$.
To see that the (dis)satisfactions of the voters coincide in both models, observe that for a voter $i$ and allocation $a \in \D$, the following holds:
$
    \ellone{p_i}{a} = 2b - 2 \sum_{j \in [m]} \min(p_{i,j}, a_j).
$
This is equal to $2b - 2 |A_i \cap W|$, so a voter $i$ prefers an allocation $a$ over another allocation $a'$ if and only if voter $i$ prefers the corresponding committee $W$ over $W'$.

\begin{figure}
    \centering
    
    \begin{tikzpicture}[xscale=1.25, yscale=0.8]
        \def\yaxpos{0.5}
        \def\coorddist{1.4}
        \def\coordwidth{0.4}
        \def\b{4}
        \def\m{3}
        \def\allone{3}
        \def\alltwo{1}
        \def\allthree{0}
        \draw (\yaxpos,0) -- (\yaxpos,\b);  
        \draw (\yaxpos-0.16,0) -- (\yaxpos+0.16,0); 
        \draw (\yaxpos-0.16,\b) -- (\yaxpos+0.16,\b); 
        \foreach \i in {2,...,\b}{
            \draw (\yaxpos-0.08,\i-1) -- (\yaxpos+0.08,\i-1); 
        }
        \node[anchor=east, xshift=-5px] at (\yaxpos,0) {$0$}; 
        \node[anchor=east, xshift=-5px] at (\yaxpos,\b) {$\b$}; 
        \node[anchor=south, rotate=90] at (\yaxpos-0.15,\b/2) {Budget}; 
        \foreach \i in {1,...,\m}{
            \filldraw[fill=black!15,draw=none] (\i*\coorddist-\coordwidth,0) rectangle (\i*\coorddist+\coordwidth,\b); 
        }
        \foreach \i in {0,...,\b}{
            \draw [dotted, darkgray] (\yaxpos,\i) -- (\m*\coorddist+\coordwidth+0.2, \i); 
        }
        \node[anchor=north east, outer sep=4.4pt] at (1*\coorddist,0) {Alternatives:};
        \foreach \i in {1,...,\m}{
            \node[anchor=north, outer sep=5pt] at (\i*\coorddist,0) {$\i$};
        }
        \foreach \all [count=\j] in {\allone, \alltwo, \allthree}{
            \foreach \k in {1, ..., \b}{
                \ifthenelse{\not{\all<\k}}{
                    \filldraw[fill=color1!50, draw=gray] (\j*\coorddist-\coordwidth,\k-1) rectangle (\j*\coorddist+\coordwidth,\k) node [midway]{$c_{\j,\k}$}; 
                }{
                    \filldraw[fill=white, draw=gray] (\j*\coorddist-\coordwidth,\k-1) rectangle (\j*\coorddist+\coordwidth,\k) node [midway]{$c_{\j,\k}$}; 
                }
             }
        }
        \foreach \all [count=\j] in {\allone, \alltwo, \allthree}{
            \draw[line width=2pt, color2] (\j*\coorddist-\coordwidth,\all) -- (\j*\coorddist+\coordwidth,\all) node [anchor=west, xshift=-3pt]{$a_{\j}$}; 
        }
    \end{tikzpicture}
    \caption{
        Example showing for $m=3$ and $b=4$ how a vote $p_i \in \D$ can be interpreted as an approval ballot, i.e., $p_i = (3,1,0)$ is translated to $A_i = \{c_{1,1}, c_{1,2}, c_{1,3}, c_{2,1}\}$. We apply a similar translation when mapping an allocation $a$ to a committee $W$.
    }
    \label{fig:approval_voting}
\end{figure}
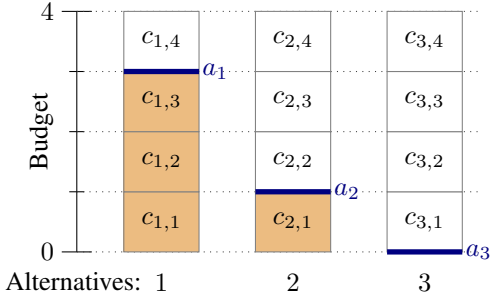

Using this connection to approval-based committee voting, we translate two representation axioms to our setting. 

\paragraph[Justified Representation]{Justified Representation (JR) \citep{ABC+15a}.} For a profile $P = (p_1, \dots, p_n)$, we say that a voter group $N' \subseteq [n]$ is \textit{cohesive} if $|N'| \ge \frac{n}{b}$ and, for some alternative $j$, it holds that $p_{i,j} > 0$ for all $i \in N'$. An allocation $a \in \D$ provides \textit{JR} if for each cohesive group $N'\subseteq [n]$, there is a voter $i \in N'$ and an alternative $j$ such that $a_j > 0$ and $p_{i,j} > 0$. 
A mechanism provides JR if it always returns an allocation providing JR. 

\paragraph[Extended Justified Representation +]{Extended Justified Representation+ (EJR+) \citep{BrPe23a}.} For a profile $P = (p_1, \dots, p_n)$, an allocation $a \in \D$ provides \textit{EJR+} if there is no alternative $j$, integer $\ell \in [b]$, and voter group $N' \subseteq [n]$ with $|N'| \ge \ell \cdot \frac{n}{b}$ such that $p_{i,j} > a_j$ and $\sum_{j' \in [m]} \min(p_{i,j'}, a_{j'}) < \ell$ for all voters $i \in N'$.
A mechanism provides EJR+ if it always returns an allocation providing EJR+.

We establish an impossibility result for each of these axioms. 
For the first impossibility, we need the additional axiom \textit{anonymity}, which disallows a mechanism from making decisions based on the identity of the voters.
(However, a mechanism can still discriminate among the alternatives.)

\paragraph{Anonymity}
 A mechanism $\A$ is \emph{anonymous} if for any profile $(p_1, \dots, p_n)$ and any permutation of voters $\sigma: [n] \rightarrow [n]$, it holds that
    $
        \A(p_1, \dots, p_n) = \A(p_{\sigma(1)}, \dots, p_{\sigma(n)}).
    $

\begin{restatable}{theorem}{thmJR}\label{thm:JR}
    No integral mechanism satisfies anonymity, truthfulness, and JR.
\end{restatable}

In order to prove \Cref{thm:JR}, we use a computer-aided approach similar to the ones used, e.g., by \citet{Pete18c}, \citet{BBP+21a}, and \citet{DDE+23b}. 
For fixed $n,m,b$, we translate the search for an anonymous, truthful, and JR mechanism into a SAT formula, and use a SAT-solver to check for satisfiability. Each satisfying assignment corresponds to a mechanism $\A_{n,m,b}$ satisfying these axioms. For $n = 3$, $m = 4$, and $b = 3$, the SAT formula is unsatisfiable, which implies that no anonymous, truthful, and JR mechanism exists. We explain how to encode these axioms into a SAT problem and give a full proof of \Cref{thm:JR} in \Cref{app:truthful_proportional}. We extracted a proof that is human-readable, but lengthy---it argues over 45 different profiles and applies truthfulness 203 times.
Therefore, we additionally present a second result with a (much) shorter proof. For this result, we consider the stronger proportionality notion EJR+ and add range-respect to the list of axioms. In return, this impossibility does not require anonymity as one of the axioms.

\paragraph{Range-respect.} A mechanism $\A$ is \emph{range-respecting} if for any $n,m,b$ and any profile $P = (p_1, \dots, p_n) \in \Dc$, the following holds for the allocation $a = \A(P)$:
$$
    \min_{i \in [n]} p_{i,j} \le a_j \le \max_{i \in [n]} p_{i,j}  \text{ for all } j \in [m].
$$

\begin{restatable}{theorem}{thmEJRplusRR}\label{thm:EJRplus_RR}
    No integral mechanism satisfies truthfulness, EJR+, and range-respect.
\end{restatable}

\begin{proof}
    Suppose for contradiction that there is a truthful, EJR+, and range-respecting mechanism $\A$. 
    Let $n=3$, $m=4$, and $b=3$, and
    consider the profile
    \begin{align*}
        P &= ((1,2,0,0),(1,0,2,0),(1,0,0,2)).
    \end{align*}
    Range-respect requires the first alternative to receive exactly $1$, leaving alternative $2$, $3$, or $4$ with zero budget. 
    Assume first that $\A(P)_2 = 0$. Consider the profile 
    \begin{align*}
        P^\star &= ((0,3,0,0),(1,0,2,0),(1,0,0,2)). 
    \end{align*}

    We claim that EJR+ implies that $\A(P^\star)_2 \ge 1$ and $\A(P^\star)_1 \ge 1$. For the former statement, notice that otherwise the voter set $\{1\}$ yields an EJR+ violation. For showing the latter statement, we now assume $\A(P^\star)_2 \ge 1$ and suppose for contradiction that $\A(P^\star)_1 =0$. 
    If $\A(P^\star)_3 = \A(P^\star)_4 = 1$, EJR+ is violated for for alternative $1$ and voter set $\{2,3\}$. Otherwise, there is $j \in \{3,4\}$ with $\A(P^\star)_j = 0$, which violates EJR+ for alternative $j$ and voter set $\{j-1\}$.
    
    Hence, we must have $\A(P^\star)_1 \ge 1$ and $\A(P^\star)_2 \ge 1$.
    However, this contradicts truthfulness, as voter~$1$ from profile $P$ can misreport $(0,3,0,0)$ to decrease her disutility.

    The cases $\A(P)_3 = 0$ and $\A(P)_4 = 0$ can be handled similarly by choosing $((1,2,0,0),(0,0,3,0),(1,0,0,2))$ or $((1,2,0,0),(1,0,2,0),(0,0,0,3))$ for $P^\star$, respectively.
\end{proof}

\section{Fractional-Input Mechanisms} \label{sec:continuous_input}

While both the integral and fractional budget aggregation settings allow for truthful mechanisms, we demonstrate in this section that truthful fractional-input mechanisms (i.e., those that map from $\Cc$ to $\D$) are significantly more restricted. In particular, we prove that the only \textit{truthful} and \textit{onto} fractional-input mechanisms are \textit{dictatorial}. This stands in contrast to the integral setting, where one can verify that, e.g., \fIM is onto and non-dictatorial. Our result builds upon the literature on dictatorial domains in ranked-choice elections. Thus, before formalizing our result in \Cref{subsec:imp-fractional}, we briefly introduce ranked-choice elections along with a result on dictatorial domains by \citet{ACS03b}. 

\subsection{Dictatorial Domains}
Let $A$ be a set of alternatives and $\mathcal{L}(A)$ be the set of all strict rankings over $A$. We call $\mathbb{D} \subseteq \mathcal{L}(A)$ a \emph{(sub)domain}. In the following, we state the concept of \textit{linkedness} for subdomains, as defined by \citet{ACS03b}. 

\paragraph{Linked Domains.}
Let $\mathbb{D} \subseteq \mathcal{L}(A)$ be a subdomain. 
\begin{itemize}
    \item We call two alternatives $a,a' \in A$ \emph{connected} in $\mathbb{D}$ if there exist strict rankings $\triangleright, \triangleright' \in \mathbb{D}$ such that $a$ is ranked first by $\triangleright$ and second by $\triangleright'$, and vice versa for $a'$. 
    \item We say that alternative $a \in A$ is \emph{linked} to a subset $B \subseteq A$ if there exist distinct $a', a'' \in B$ such that $a$ is connected to both $a'$ and $a''$ in $\mathbb{D}$.
    \item We call the subdomain $\mathbb{D}$ \emph{linked} if we can order the alternatives in $A$ into a vector $(a^{1}, \dots, a^{|A|})$ such that 
    $a^{1}$ is connected to $a^{2}$ and, for all $k \in \{3, \dots, |A|\}$,
    it holds that $a^{k}$ is linked to $\{a^{1}, \dots, a^{k-1}\}$.
\end{itemize}

Informally, \citet{ACS03b} have shown that the Gibbard--Satterthwaite theorem \citep{gibbard1973manipulation,satterthwaite1975strategy} holds for all linked domains.
We state their theorem below and defer the formal definitions of a \textit{social choice function}, \textit{unanimous}, \textit{truthful}, and \textit{dictatorial} in the context of ranked-choice voting to \Cref{sec:strengthening}. 

\begin{theorem}[{\cite[Theorem 3.1]{ACS03b}}] \label{thm:strict_order_impossibility}
    For any set of alternatives $A$ with $|A|\ge 3$, the following holds: If a subdomain $\mathbb{D} \subseteq \mathcal{L}(A)$ is linked, then any unanimous and truthful social choice function on domain $\mathbb{D}$ is dictatorial for any number of voters  $n \in \mathbb{N}$.
\end{theorem}

For our proof, we need a stronger version of this theorem, which works even for weak rankings that have no ties in the two top ranks. We formalize this version and argue why it holds in \Cref{sec:strengthening}.

\subsection{Truthful Fractional-Input Mechanisms} \label{subsec:imp-fractional}

There exists a direct connection between our model and that of weak rankings. Namely, each vote $p \in \C$ induces a weak ranking $\wpo_p$ over the integral allocations in $\D$ (i.e., rank points in $\D$ by their $\ell_1$-distance to $p$). 
At a high level, our goal is therefore to show that these weak rankings form a linked domain, which together with the stronger version of \Cref{thm:strict_order_impossibility} yields a similar result in our setting. 

Before doing so, we return to the context of fractional-input mechanisms and formalize the desired result. 

\paragraph{Onto.} A fractional-input mechanism $\A$ is \emph{onto} if for any $n, m, b \in \N$ with $m \ge 2$ and any integral allocation $a \in \D$, there exists a profile $P \in \Cc$ with $\A(P) = a$.

\paragraph{Dictatorial.} Given $n,m,b\in\N$ with $m\ge 2$, voter $i \in [n]$ is a \textit{dictator} for a fractional-input mechanism $\A$ for $n,m,b$ if for all profiles $P = (p_1, \dots, p_n)$ with parameters $m$ and $b$, it holds that $\A(P)$ has rank~$1$ (i.e., is most preferred) in $\wpo_{p_i}$. 
The mechanism $\A$ is \emph{dictatorial} for $n,m,b$ if there exists a voter that is a dictator for $\A$ for $n,m,b$.

\begin{restatable}{theorem}{continuousImpossibility}\label{thm:continuous_impossibility}
     Any onto and truthful fractional-input mechanism is dictatorial for any $n,m,b$ with $(m-1) \cdot b \ge 2$. 
\end{restatable}

\begin{proof}[Proof Sketch of \Cref{thm:continuous_impossibility}]
We start by defining a set of weak rankings induced by $\C$, namely, \[\nabla = \left\{\wpo_p\, \mid p \in \C \text{ and } |r_1(\wpo_p)| = |r_2(\wpo_p)| = 1 \right\},\] where $\wpo_p$ is as defined at the beginning of \Cref{subsec:imp-fractional}, and $r_1(\wpo_p)$ (resp., $r_2(\wpo_p)$) denotes the set of alternatives ranked first (resp., second) by $\wpo_p$. 
We prove that this domain is linked, according to an adaptation of the definition of linkedness by \citet{ACS03b} to weak rankings that have singleton top ranks. 
To this end, we carefully construct a ranking of the elements in $\D$ that witnesses the linkedness of $\nabla$. 
Assume for contradiction that there exists a fractional-input mechanism $\mathcal{A}$ that is onto, truthful, and non-dictatorial for some $n \in \N$. We show that this implies the existence of a social choice function $\mathcal{B}$ on domain $\nabla$ that is unanimous, truthful, and non-dictatorial for $n$ voters, which contradicts the strengthened version of \Cref{thm:strict_order_impossibility}. While proving unanimity and truthfulness for $\B$ is rather immediate, establishing that $\B$ is non-dictatorial requires more effort. This is because $\nabla$ can be interpreted\footnote{This is not formally precise, as $\nabla$ contains rankings and not elements of $\C$.
However, we build a bijection between $\nabla$ and a subset of $\C$.} as a subdomain of $\C$, and $\A$ being non-dictatorial on $\Cc$ does not directly imply that $\B$ is non-dictatorial on $\nabla^n$. 
\end{proof}

The sharp contrast between the fractional-input and integral settings in relation to truthfulness may seem surprising. However, we remark that integral moving-phantom mechanisms can be used to construct fractional mechanisms that are approximately truthful, and the incentive to misreport diminishes as the budget increases. Specifically, we map each vote $p \in \C$ to a point in $\D$ closest to it (with $\ell_1$-distance at most $\frac{m}{2}$) and apply an integral moving-phantom mechanism. By the triangle inequality, the disutility decrease from misreporting is bounded by $2 \cdot \frac{m}{2} = m$. Thus, for constant $m$, (relative) misreporting incentives vanish as $b$ grows.

\section{Conclusion and Future Work}

In this paper, we have introduced the setting of discrete budget aggregation, which reflects the integrality requirement on the output often found in budget aggregation applications, and studied it with respect to truthfulness and proportionality axioms.
Regarding truthfulness, we established a sharp contrast between the integral and the fractional-input settings: in the former, we presented a class of truthful mechanisms by building upon the literature on fractional budget aggregation, while in the latter, we exhibited the limitations of truthful mechanisms by leveraging existing results on dictatorial domains. 
Regarding proportional representation, we demonstrated that our integral setting can be interpreted as a subdomain of approval-based committee voting, but even basic representation guarantees from this literature are incompatible with truthfulness. In contrast, we proved that proportionality can be attained when voters are single-minded. 

Our paper leaves several intriguing directions for future work. 
First, it would be useful to characterize the class of truthful integral mechanisms. 
For the fractional setting, \citet{deberg2024truthful} have recently shown that there exist truthful mechanisms beyond moving-phantom mechanisms. 
While characterizing all truthful mechanisms appears to be difficult in the fractional case given that some of these mechanisms are arguably unnatural, the question may be easier to answer in the integral case. 
Another interesting avenue is to further explore the connections of budget aggregation to approval-based committee voting, independently of truthfulness. For example, which mechanisms do we obtain in the fractional setting if we apply well-established committee rules, such as the \emph{method of equal shares} \citep{peters2020proportionality}, in the integral setting and let the budget approach infinity? 

\section*{Acknowledgments}
This work was partially 
supported by the Dutch Research Council (NWO) under project number VI.Veni.232.254, by the Singapore Ministry of Education under grant number MOE-T2EP20221-0001, and by an NUS Start-up Grant. We thank the anonymous reviewers for their valuable feedback.

\bibliographystyle{named}
\bibliography{literature}

\newpage
\appendix

\onecolumn 

\section{Missing Definitions and Proofs from \Cref{sec:truthful}}
\label{app:integral-truthful}

We start by formally defining some apportionment methods. To account for tie-breaking, we define apportionment methods in a non-resolute fashion in the following. 
An apportionment method maps an element from \C to a subset of elements in \D, that is, $\M: \C \rightarrow 2^{\D}$. Then, a tie-breaking rule $\beta$ selects one of the outcomes, i.e., for any $S \subseteq \D$, we have $\beta(S) \in S$. 
Hence, formally, we study the composition $\beta \circ \M \circ \A$. We now provide the definitions of different apportionment methods.

\paragraph{Hamilton's Method} Given a vector $a \in \C$, let $r$ be the vector of residues, i.e., $r_j = a_j - \lfloor a_j \rfloor$ for each $j \in [m]$. Note that $k := \sum_{j \in [m]} r_j$ is an integer. 
Hamilton's method first gives every alternative $j$ at least $\lfloor a_j \rfloor$. Moreover, it gives $\lceil a_j \rceil$ to the $k$ alternatives maximizing $r_j$. Since there may be multiple alternatives with the same residue, Hamilton returns one vector per subset of $k$ alternatives maximizing $r_j$. 

\paragraph{Stationary Divisor Methods} Stationary divisor methods are parameterized by $\Delta \in [0,1]$, which defines a rounding function. Formally, 
\[
    \llbracket z \rrbracket_\Delta = \begin{cases}
        \{y \} & \text{ if } y-1 + \Delta < z < y + \Delta \text{ for some } y\in \mathbb{N}\cup\{0\},\\
        \{y,y+1\} & \text{ if } z = y + \Delta \text{ for some } y\in \mathbb{N} \cup \{0\}.
    \end{cases}
\]

For $\Delta \in [0,1]$, the \textit{$\Delta$-divisor method} $\M_{\Delta}$ is defined by
\[
    \M_{\Delta}(a) = \bigg\{ x\in (\mathbb{N}\cup\{0\})^n ~\bigg|~ \text{there exists } \lambda>0 \text{ such that\ } x_j \in \llbracket \lambda a_j\rrbracket_\Delta \text{ for every }j\in [m] \text{ and } \sum_{i=1}^{m} x_j = \sum_{i=1}^{m} a_j \bigg\}.
\]
The Jefferson/D'Hondt method corresponds to $\M_{1}$ in our notation, and the Webster/Sainte-Lagu\"e method to $\M_{0.5}$.

\paragraph{Quota Method} The quota method can be seen as a constrained version of (a sequential interpretation of) Jefferson's method. That is, we assign the budget of $b$ iteratively. For each round $k \in [b]$, let $\gamma_j$ be the current budget of alternative $j$. Then, in each round we choose an alternative $j$ minimizing $(\gamma_j+1)/a_j$ over a restricted subset of ``eligible'' alternatives. An alternative is \emph{eligible} if allocating the next unit of budget to it would not give it more than its ``fair share'' rounded up. Formally, in round~$k$, the set of eligible agents is $U(a,\gamma,k) = \left\{j \in [m]\,\middle|\, \gamma_j < \frac{a_j k }{ \sum_{j' \in [m]} a_{j'}}\right\}$, where $\gamma=(\gamma_1,\dots,\gamma_m)$. 
Among all eligible alternatives, the next unit of budget is assigned to an alternative $j$ minimizing $(t_j+1)/a_j$. The output of the quota method consists of all vectors that can be achieved by some way of breaking ties. 

\thmApportionmentNoTie*

\begin{proof}
    \textbf{Hamilton's method.} Consider the following preference profile with $n =10$, $m=6$, and $b=8$: 
        $$ P = 
        \begin{bmatrix}
            8 & 0 & 0 & 0 & 0 & 0  \\
            8 & 0 & 0 & 0 & 0 & 0  \\
            8 & 0 & 0 & 0 & 0 & 0  \\
            8 & 0 & 0 & 0 & 0 & 0  \\
            8 & 0 & 0 & 0 & 0 & 0  \\
            7 & 1 & 0 & 0 & 0 & 0  \\
            7 & 0 & 1 & 0 & 0 & 0  \\
            7 & 0 & 0 & 1 & 0 & 0  \\
            7 & 0 & 0 & 0 & 1 & 0  \\
            7 & 0 & 0 & 0 & 0 & 1  \\
        \end{bmatrix},
        $$
    where every row corresponds to the report of one voter. For this profile, \textsc{IndependentMarkets} returns $(\sfrac{80}{15},\sfrac{8}{15},\sfrac{8}{15},\sfrac{8}{15},\sfrac{8}{15},\sfrac{8}{15})$. As a result, Hamilton's method rounds up three of the alternatives $2,3,4,5,6$; without loss of generality, we can assume that the outcome of Hamilton's method is $(5,1,1,1,0,0)$. 
    Now, suppose that the last voter changes her report to $(8,0,0,0,0,0)$. For the new profile, \IM returns $(\sfrac{40}{7},\sfrac{4}{7},\sfrac{4}{7},\sfrac{4}{7},\sfrac{4}{7},0)$, and Hamilton's method rounds up the first alternative along with two of the alternatives $2,3,4,5$ (this leads to, e.g., the outcome $(6,1,1,0,0,0)$). Any of these outcomes is strictly preferred to the original outcome by the last voter, so truthfulness is violated. 

    \bigskip 
    \textbf{Quota Method.} Consider the following two profiles with $n=b=4$ and $m =5$: 
        $$
            P = \begin{bmatrix}
            3 & 1 & 0 & 0 & 0\\
            3 & 0 & 1 & 0 & 0\\
            3 & 0 & 0 & 1 & 0\\
            3 & 0 & 0 & 0 & 1\\
            \end{bmatrix} \qquad
            P^\star = \begin{bmatrix}
            3 & 1 & 0 & 0 & 0\\
            3 & 0 & 1 & 0 & 0\\
            3 & 0 & 0 & 1 & 0\\
            4 & 0 & 0 & 0 & 0\\
            \end{bmatrix}.
        $$
        \IM returns $a = (2,0.5,0.5,0.5,0.5)$ for profile $P$, which is (without loss of generality) rounded to $a=(2,1,1,0,0)$ by the quota method. For the profile $P^{\star}$, \IM returns $(\sfrac{16}{7},\sfrac{4}{7},\sfrac{4}{7},\sfrac{4}{7},0)$. Note that the first alternative is eligible to receive a budget of $3$ and that $\frac{16}{7}\cdot \frac{1}{3} > \frac{4}{7}$. Hence, the quota method returns (without loss of generality) $(3,1,0,0,0)$, which is strictly preferred by the last voter.
        This causes a violation of truthfulness.

    \bigskip
    
    \textbf{Stationary Divisor Methods.} Consider the integral profile $P \in \Dc$ with $m = \lceil 2+\frac{2}{\Delta} \rceil$ alternatives, $n = m$ voters, and a budget of $b = 2$, where for each $i \in [n-1]$ we have $p_{i,1} = p_{i,i+1} = 1$, and $p_n = (1,1,0,\dots,0)$. As an example, we display this profile for $\Delta = 1$ (so $n = m = 4$):
    $$
        \begin{bmatrix}
            1 & 1 & 0 & 0  \\
            1 & 0 & 1 & 0  \\
            1 & 0 & 0 & 1  \\
            1 & 1 & 0 & 0  \\
        \end{bmatrix}.
    $$
    \textsc{IndependentMarkets} reaches normalization at time $t = \frac{1}{2n}$, when phantom $k$ is at position $\frac{n-k}{n}$ for each $k \in [n]_0$. 
    Hence, the output of \textsc{IndependentMarkets} is $(1, \frac{2}{n}, \frac{1}{n}, \dots, \frac{1}{n})$. Taking this as an input to the stationary divisor method $\M_{\Delta}$, we note that choosing the multiplier $\lambda = 1+\Delta$ leads to $1\lambda = 1+\Delta$ and $\frac{1}{n} \lambda < \frac{2}{n} \lambda \leq \frac{2}{(2 + \frac{2}{\Delta})(\frac{1}{1+\Delta})} = \Delta$. Here, note that the weak inequality is an equality if and only if $\frac{2}{\Delta} \in \mathbb{N}$; in this case, $\M_{\Delta}$ may return $(2,0,\dots,0)$ or $(1,1,0,\dots,0)$, and we are going to assume that the tie-breaking rule chooses in favor of $(2,0,\dots,0)$. If $\frac{2}{\Delta} \not\in \mathbb{N}$, then the unique output is $(2,0,\dots,0)$. 

    Now, consider the profile $P^{\star}$ in which the last voter misreports $p_n^\star = (0,2,0,\dots,0)$. \textsc{IndependentMarkets} reaches normalization at time $t = \frac{1}{2n-1}$, when phantom $k$ is at position $\frac{2(n-k)}{2n-1}$ for each $k \in [n]_0$. The output of \IM is therefore $x = (\frac{2(n-1)}{2n-1}, \frac{4}{2n-1}, \frac{2}{2n-1}, \dots, \frac{2}{2n-1})$, for which we can find the multiplier $\lambda^\star = \frac{\Delta(2n-1)}{4}$ such that $x_1 \lambda^\star = \frac{(n-1)\Delta}{2} = (\lceil 2 + \frac{2}{\Delta} \rceil -1)\frac{\Delta}{2} < (2+\frac{2}{\Delta})\frac{\Delta}{2} = 1+\Delta$ and $x_2 \lambda^\star = \Delta$. Thus, the output of $\M_{\Delta}$ is $(1, 1, 0, \dots, 0)$, which is strictly preferred by voter $n$.
    This yields a contradiction to truthfulness. 
\end{proof}

\thmDiscreteTruthful*
\begin{proof}
    The proof is analogous to the proof of truthfulness of fractional moving-phantom mechanisms given by \citet[Theorem~2]{freeman2021truthful}. Let $m, n, b \in \N$ be fixed, $P = (p_1, \dots, p_n) \in \Dc$ an integral profile, and $P^\star = (p_1, \dots, p_{i-1}, p_i^\star, p_{i+1}, \dots, p_n)$ a profile with a misreport $p_i^\star \neq p_i$ by voter $i$. Let $\A^\W$ be a mechanism defined by the integral phantom system $\W$, and let $a = \A^\W(P)$ and $a^\star = \A^\W(P^\star)$. Further, let $\tau$ and $\tau^\star$ be times at which $\sum_{j \in [m]} \med(P,\W,j,\tau) = b$ and $\sum_{j \in [m]} \med(P^\star,\W,j,\tau^\star) = b$, respectively. 
    
    We first show that any change of the medians resulting from reporting $p_i^\star$ instead of $p_i$ can only increase the disutility of voter $i$ if we fix all phantoms $\ph_{k,j}$ at their positions $\ph_{k,j}(\tau)$. We then show that updating the phantom positions from $\ph_{k,j}(\tau)$ to $\ph_{k,j}(\tau^\star)$ can result in a disutility decrease of at most the amount that it increased before.
    
    Let us first consider the medians on each alternative of the phantom positions $\ph_{k,j}(\tau)$ and the profile $P$. For each alternative $j$, the median only changes if the voter ``crosses'' it by moving from $p_i$ to $p_i^\star$, i.e., we have
    \begin{align*}
        \med(P^\star,\W,j,\tau) < \med(P,\W,j,\tau) \qquad & \text{ if } p^\star_{i,j} < \med(P,\W,j,\tau) \le p_{i,j}, \\ 
        \med(P^\star,\W,j,\tau) > \med(P,\W,j,\tau) \qquad & \text{ if } p_{i,j} \le \med(P,\W,j,\tau) < p^\star_{i,j}, \\ 
        \med(P^\star,\W,j,\tau) = \med(P,\W,j,\tau) \qquad & \text{ otherwise}. 
    \end{align*}
    Therefore, any change of the medians will be in the direction away from $p_i$, thus increasing voter $i$'s disutility by
    \begin{align}
        y 
        & := \sum_{j\in[m]} |p_{i,j} - \med(P^\star, \W, j, \tau)| - \sum_{j\in[m]} |p_{i,j} - \med(P, \W, j, \tau)| \label{eq:disutility_change_step1}\\
        & = \sum_{j\in[m]} |\med(P^\star, \W, j, \tau) - \med(P, \W, j, \tau)|.\nonumber
    \end{align}
    
    We now consider how the disutility can change when updating the phantom positions from $\tau$ to $\tau^\star$. Assume that $\tau \le \tau^\star$ and thus $\sum_{j\in[m]} \med(P^\star, \W, j, \tau) \le b$; the proof for $\tau \ge \tau^\star$ works analogously. Since for all $k,j$ we know that $\ph_{k,j}$ is non-decreasing, the medians $\med(P^\star, \W, j, \tau)$ must be lower than the medians $\med(P^\star, \W, j, \tau^\star)$ by a total amount of
    \begin{align}
        \begin{split}
            \sum_{j\in[m]} |\med(P^\star, \W, j, \tau^\star) - \med(P^\star, \W, j, \tau)|
            & = b - \sum_{j\in[m]} \med(P^\star, \W, j, \tau) \\
            & = \sum_{j\in[m]} \med(P, \W, j, \tau) - \sum_{j\in[m]} \med(P^\star, \W, j, \tau) \\
            & \le \sum_{j\in[m]} |\med(P, \W, j, \tau) - \med(P^\star, \W, j, \tau)| = y. \label{eq:disutility_change_step2}
        \end{split}
    \end{align}

    Thus, for the total disutility of voter $i$, we get
    \begin{align*}
        \ellone{p_i}{\A^\W(P^\star)} 
        & =   \sum_{j\in[m]} |p_{i,j} - \med(P^\star, \W, j, \tau^\star)| \\
        \text{\scriptsize(triangle inequality)} 
        & \ge \sum_{j\in[m]} \big( |p_{i,j} - \med(P^\star, \W, j, \tau)| - |\med(P^\star, \W, j, \tau) - \med(P^\star, \W, j, \tau^\star)| \big)\\
        & = \sum_{j\in[m]} |p_{i,j} - \med(P^\star, \W, j, \tau)| - \sum_{j\in[m]} |\med(P^\star, \W, j, \tau) - \med(P^\star, \W, j, \tau^\star)|\\
        \text{\scriptsize((\ref{eq:disutility_change_step1}) and (\ref{eq:disutility_change_step2}))}
        & \ge \left(\sum_{j\in[m]} |p_{i,j} - \med(P, \W, j, \tau)| + y\right) - y\\
        & = \ellone{p_i}{\A^\W(P)},
    \end{align*}
    which concludes the proof.
\end{proof}

\floorUtilEqualsUtilHamilton*

\begin{proof}
    For any $n,m,b \in \N$, consider a profile $P$ and let $\F_n = \{f_0, \dots, f_n\}$ be the phantom system of the \textsc{Utilitarian} mechanism. Let $\W_{n,m,b} = \{\ph_{k,j} \mid k \in [n]_0,\, j \in [m]\}$ be the integral phantom system induced by $\F_n$ and the floor function. We denote Hamilton's method by $\M$.
    
    Let $a = \A^\F(P)$ be the fractional allocation returned by \textsc{Utilitarian}, let $t$ be a time of normalization of $\A^\F$ for $P$, and let $\tau$ be the time step where $\ph_{k,j}(\tau) = \lfloor f_k(t)\rfloor$ for $k \in [n]_0$ and $j \in [m]$.
    Denote by $J \subseteq [m]$ the set of alternatives for which $a_j$ is non-integral.
    Since all votes lie on integral points, the medians for all alternatives $j \in J$ must lie on phantom positions. 
    Since the utilitarian mechanism moves the phantoms in sequence, only one phantom $k^\star$ can be at a non-integral position at time $t$, and we have $a_j = f_{k^\star}(t)$ for all $j \in J$. 

    Let $x = b - \sum_{j \in [m]} \med(P, \W, j, \tau)$ be the amount of budget missing for normalization at time step $\tau$. 
    Now, let $t'$ be the time at which phantom $k^\star$ reaches $\lceil f_{k^\star}(t) \rceil$ (note that all other phantoms did not move relative to time $t$). By the construction of $\W$, the discrete phantoms $\ph_{j,k^\star}$ increase their position in order of $j$. 
    For the aggregate $\hat{a} = \A^\W(P)$, we therefore have that $\hat{a}_j = a_j$ for $j \notin J$, $\hat{a}_j = \lceil a_j \rceil$ for the $x$ smallest $j \in J$, and $\hat{a}_j = \lfloor a_j \rfloor$ for all other $j \in J$. 

    Similarly, for Hamilton's method, since all non-integral values in $a$ are the same, the decision on which alternatives are rounded up and which are rounded down is entirely decided by the tie-breaking. Since we assume tie-breaking by the indices of the alternatives, this leads to the same aggregate $\M(a) = \hat{a}$.
\end{proof}

\section{Missing Proofs from \Cref{sec:truthful_proportional}} \label{app:truthful_proportional}

We first prove \Cref{thm:quota_proportionality}, which establishes the existence of single-minded quota-proportional mechanisms.

\thmQuotaProportionality*
\begin{proof}
    Let $\F$ be a family of upper-quota capped phantom systems that induce a single-minded proportional moving-phantom mechanism $\A^\F$. Let $\W$ be the family of integral phantom systems induced by $\F$ and the floor function, according to the approach described at the end of \Cref{subsec:rounding_phantoms}. We want to show that $\A^\W$ satisfies single-minded quota-proportionality. Let $n,m,b \in \N$ and $P \in \Dc$ be a single-minded profile. For each $j\in[m]$, let $n_j$ be the number of voters $i$ with $p_{i,j} = b$. The average on each alternative $j$ is then given by $\Amean(P)_j = b \cdot \frac{n_j}{n}$.
    We prove that $\A^\W$ satisfies single-minded quota-proportionality by showing that it outputs an aggregate $a$ with $a_j \in \{\lfloor b \cdot \frac{n_j}{n} \rfloor, \lceil b \cdot \frac{n_j}{n} \rceil\}$ for each alternative $j$. 
    
    Since on each alternative $j$ there are $n_j$ voters at position $b$ and $n-n_j$ voters at position $0$, the median at any time step $\tau$ equals the $(n-n_j+1)$th phantom of $\W$, i.e., $\med(P, \W, j, \tau) = \ph_{n-n_j,\,j}(\tau)$. Similarly, for the fractional phantom system $\F$, the median on alternative $j$ is $\med(P, \F, j, t) = f_{n-n_j}(t)$ at any time $t$. 

    By the single-minded proportionality of $\A^\F$, we know that at any time of normalization $t$ for profile $P$, we have $f_{n-n_j}(t) = b \cdot \frac{n_j}{n}$ for all alternatives $j$. 
    By construction of $\phi$, there must exist a time step $\tau_1 \in \N$ at which $\ph_{n-n_j,\,j}(\tau_1) = \lfloor b \cdot \frac{n_j}{n} \rfloor$ on each alternative $j$. Then, the sum of medians at that time is
    $$
        \sum_{j\in[m]} \med(P, \W, j, \tau_1) = \sum_{j\in[m]} \ph_{n-n_j,\, j}(\tau_1) = \sum_{j\in[m]}  \left\lfloor b \cdot \frac{n_j}{n}\right\rfloor \le \sum_{j\in[m]} b \cdot \frac{n_j}{n} = b,
    $$
    showing that the aggregate $a_j$ must be at least $\lfloor b \cdot \frac{n_j}{n}\rfloor$ on each alternative $j$.

    Similarly, since $\F$ is upper-quota capped, the phantoms of $\F$ reach  $(\lceil b \cdot \frac{(n-0)}{n}\rceil, \dots, \lceil b \cdot \frac{(n-n)}{n}\rceil)$ at time $1$ for each alternative $j$. Therefore, there must exist a time step $\tau_2 \in \N$ at which $\ph_{n-n_j,\,j}(\tau_2) = \lceil b \cdot \frac{n_j}{n} \rceil$ on each alternative $j$. Thus,
    $$
        \sum_{j\in[m]} \med(P, \W, j, \tau_2) = \sum_{j\in[m]} \ph_{n-n_j,\, j}(\tau_2) = \sum_{j\in[m]}  \left\lceil b \cdot \frac{n_j}{n}\right\rceil \ge \sum_{j\in[m]} b \cdot \frac{n_j}{n} = b,
    $$
    showing the upper bound $a_j \le \lceil b \cdot \frac{n_j}{n}\rceil$.
\end{proof}

We now turn towards proving \Cref{thm:JR}, which shows the incompatibility of anonymity, truthfulness, and justified representation for integral mechanisms. For constructing the proof, we use a computer-aided approach, encoding an anonymous, truthful, and JR mechanism into a SAT formula and then checking satisfiability using a SAT-solver. Below, we go into more detail on how this works. The process is similar to proofs by, e.g., \citet{Pete18c}, \citet{BBP+21a}, and \citet{DDE+23b}. 

Fix $n,m,b$. 
For $P \in \Dc$, let $\Gamma(P) \subseteq \D$ denote all allocations satisfying JR with respect to $P$. For each pair of profile $P \in \Dc$ and JR allocation $a \in \Gamma(P)$, we define a boolean variable $x_{P,a}$ representing whether $a$ is chosen as the output allocation for $P$. Since any mechanism must choose exactly one allocation per profile, we add the following clauses to guarantee that at least one and at most one variable can be positive for a profile, respectively. 
\begin{align*}
    \bigvee_{a \in \Gamma(P)} x_{P,a} & \qquad  \text{ for all } P \in \Dc \\
    \neg x_{P,a} \vee \neg x_{P,a'} & \qquad \text{ for all } P \in \Dc  \text{ and } a, a' \in \Dc \text{ with } a \neq a'\\
\end{align*}
Since we require anonymity, we can significantly reduce the number of variables by identifying all profiles that are permutations of each other.
To implement truthfulness, we add a clause for each combination of two profile-outcome pairs that violate truthfulness. 
\begin{align*}
    \neg x_{P,a} \vee \neg x_{P^\star,a^\star} & \quad \text{ for all } P = (p_1, \dots, p_n) \in \Dc \\ 
    & \quad \text{ for all } i \in [n],\, p_i^\star \in \D  \text{ with } \ellone{p_i}{a} > \ellone{p_i}{a^\star},\\
    & \quad \text{ where } P^\star = (p_1, \dots, p_{i-1}, p_i^\star, p_{i+1}, \dots, p_n)\\
\end{align*}

For $m=4$ alternatives, $n=3$ voters, and a budget of $b=3$, the SAT formula is unsatisfiable, which establishes \Cref{thm:JR}. To present a human-readable proof, we extracted a minimum unsatisfiable set of clauses and then identified a case distinction over the outcome of a specific profile that leads to a contradiction by a series of truthfulness applications.

\thmJR*

\begin{proof}

    \newcommand{\hide}[1]{\textcolor{gray}{#1}}
    \newcommand{\highlightOne}[1]{\underline{#1}}
    \newcommand{\highlightTwo}[1]{\textbf{#1}}
    
    \newcommand{\lineid}[1]{s#1}

    Suppose for contradiction that there is a mechanism $\A$ satisfying anonymity, truthfulness, and JR. Let $n = 3$, $m = 4$, and $b = 3$. We show the theorem by computing all JR outcomes for a number of specific profiles and excluding the possible outputs of $\A$ further by applying truthfulness to pairs of profiles. 
    Note that since $n = b$, JR is equivalent to the condition that for every voter, there exists an alternative such that both the voter and the outcome allocate a positive amount to it.
    For the sake of readability, we write profiles and allocations without parentheses and commas, e.g., the profile $((0,0,2,1), (0,3,0,0), (0,0,0,3))$ is written as 0021~0300~0003.
    Since $\A$ satisfies anonymity, it returns the same outcome for all permutations of a profile. Thus, we represent each profile by its lexicographically smallest permutation, e.g., we represent all permutations of the profile 0021~0300~0003 by the profile 0003~0021~0300. 
    
    The following tables can be read as follows:
    For each profile in the table, we list all JR outcomes in the ``Outcomes'' column. Outcomes that have been excluded by earlier arguments are grayed out. Outcomes that are newly excluded by the current argument are grayed out and underlined. 
    Each line in the table thus represents a \textit{statement}, which says that the outcome of $\A$ for the given profile must be one of the black allocations. Each unique statement has an ID. 
    
    Every pair of adjacent lines (without whitespace) represents an \textit{argument}, which corresponds to a truthfulness application. The profiles from a pair of adjacent lines only differ in one vote and we can exclude outcomes of the \textit{second} profile by truthfulness. For this, we need to assume that one of the profiles has the truthful vote and the other one has the misreport of the changing voter. Which profile has the truthful vote or misreport is denoted in the columns ``Truthful ID'' and ``Misreport ID''.
    To make it easier to follow the argument for a specific profile, the ``Last Mentioned'' column refers to the statement ID that last restricted the output set of the profile. 

    As an example, consider the following argument.
    
    \begin{longtable}{@{\hspace{-2mm}}>{\centering\arraybackslash}p{1cm} l p{8cm} >{\centering\arraybackslash}m{1.6cm} >{\centering\arraybackslash}m{1.2cm} >{\centering\arraybackslash}m{1.2cm}@{\hspace{2mm}}}
        \hline
        ID & Profile & Outcomes & Last Mentioned & Truthful ID & Misreport ID \\
        \hline
        \\
        \lineid{032}   &  0003 0030 2010  &   \hide{0012} \hide{0021} \hide{0111} 1011                      & & &\\
        \lineid{039}   &  0030 1002 2010  &   \hide{\highlightOne{0012}} \hide{0021} \hide{0111} 1011 \hide{1020} \hide{1110} \hide{2010}  &\lineid{038}    & \lineid{032} & \lineid{039}\\
        \\
        \hline
    \end{longtable}

    In the first line, we can see that out of all JR outcomes for the profile 0003~0030~2010, all but one have been excluded by previous arguments. Thus, we know at this point that $\A$ must return 1011, which allows us to exclude some of the outcomes of the profile in the second line. Notice that apart from the order of the voters, the two profiles only differ in one vote, which is 0003 in the first profile and 1002 in the second profile. For this argument, the table indicates that the profile from the first line (\lineid{032}) contains the truthful vote and the second one (\lineid{039}) has the misreport. Since a voter with truthful vote 0003 has disutility $4$ for the outcome 1011 and the disutility cannot decrease by misreporting, we can exclude outcome 0012 for the second profile, as this would lead to a disutility of $2$. This outcome is grayed out and underlined. The other grayed-out outcomes have been excluded by earlier arguments, which can be checked in statement \lineid{038}. 
    
    Note that there are situations in which not all outcomes that could be excluded by a truthfulness application are excluded. These outcomes will then be excluded using a different argument at a later time.
    
    \begin{longtable}{@{\hspace{-2mm}}>{\centering\arraybackslash}p{1.2cm} l p{8cm} >{\centering\arraybackslash}m{1.6cm} >{\centering\arraybackslash}m{1.2cm} >{\centering\arraybackslash}m{1.2cm}@{\hspace{2mm}}}
    \hline
    ID & Profile & Outcomes & Last Mentioned & Truthful ID & Misreport ID \\
    \hline
    \\
    \lineid{001}   &  0003 0030 0300  &   0111                                                          & & &\\
    \lineid{002}   &  0003 0021 0300  &    \hide{\highlightOne{0102}} 0111 \hide{\highlightOne{0201}} \hide{\highlightOne{1101}}  & & \lineid{002} & \lineid{001}\\
    \lineid{003}   &  0003 0021 0210  & \hide{\highlightOne{0012}} 0021 0102 0111 0201 \hide{\highlightOne{1011}} 1101  & & \lineid{003} & \lineid{002}\\
    \\
    \lineid{002}   &  0003 0021 0300  &   \hide{0102} 0111 \hide{0201} \hide{1101}                      & & &\\
    \lineid{004}   &  0012 0021 0300  &   0102 0111 \hide{\highlightOne{0120}} \hide{\highlightOne{0201}} \hide{\highlightOne{0210}} \hide{\highlightOne{1101}} \hide{\highlightOne{1110}}  & & \lineid{004} & \lineid{002}\\
    \\
    \lineid{002}   &  0003 0021 0300  &   \hide{0102} 0111 \hide{0201} \hide{1101}                      & & &\\
    \lineid{005}   &  0012 0021 0300  &   \hide{\highlightOne{0102}} 0111 \hide{0120} \hide{0201} \hide{0210} \hide{1101} \hide{1110}   &\lineid{004}    & \lineid{002} & \lineid{005}\\
    \lineid{006}   &  0012 0021 0210  &   0012 \hide{\highlightOne{0021}} 0030 0102 0111 0120 0201 0210 \hide{\highlightOne{1011}} 1020 1101 1110 2010  & & \lineid{006} & \lineid{005}\\
    \\
    \lineid{002}   &  0003 0021 0300  &   \hide{0102} 0111 \hide{0201} \hide{1101}                     & & &\\
    \lineid{007}   &  0021 0300 1002  &   \hide{\highlightOne{0102}} 0111 0201 1101 1110                 & & \lineid{002} & \lineid{007}\\
    \\
    \lineid{001}   &  0003 0030 0300  &   0111                                                          & & &\\
    \lineid{008}   &  0003 0030 0201  &   \hide{\highlightOne{0012}} \hide{\highlightOne{0021}} 0111 \hide{\highlightOne{1011}}  & & \lineid{008} & \lineid{001}\\
    \lineid{009}   &  0030 0201 1002  &   \hide{\highlightOne{0012}} 0021 0111 1011 1110                & & \lineid{008} & \lineid{009}\\
    \\
    \lineid{008}   &  0003 0030 0201  &   \hide{0012} \hide{0021} 0111 \hide{1011}                      & & &\\
    \lineid{010}   &  0030 0201 1011  &   0012 0021 0111 \hide{\highlightOne{0120}} \hide{\highlightOne{0210}} 1011 1110  & & \lineid{010} & \lineid{008}\\
    \\
    \lineid{008}   &  0003 0030 0201  &   \hide{0012} \hide{0021} 0111 \hide{1011}                      & & &\\
    \lineid{011}   &  0030 0201 1011  &   \hide{\highlightOne{0012}} 0021 0111 \hide{0120} \hide{0210} 1011 1110  &\lineid{010}    & \lineid{008} & \lineid{011}\\
    \\
    \lineid{008}   &  0003 0030 0201  &   \hide{0012} \hide{0021} 0111 \hide{1011}                      & & &\\
    \lineid{012}   &  0030 0201 1110  &   \hide{\highlightOne{0012}} \hide{\highlightOne{0021}} 0111 0120 0210 1011 1110  & & \lineid{012} & \lineid{008}\\
    \\
    \lineid{001}   &  0003 0030 0300  &   0111                                                          & & &\\
    \lineid{013}   &  0003 0030 0210  &   \hide{\highlightOne{0012}} \hide{\highlightOne{0021}} 0111 \hide{\highlightOne{1011}}  & & \lineid{013} & \lineid{001}\\
    \lineid{014}   &  0003 0021 0210  &   \hide{0012} 0021 \hide{\highlightOne{0102}} 0111 \hide{\highlightOne{0201}} \hide{1011} \hide{\highlightOne{1101}}  &\lineid{003}    & \lineid{014} & \lineid{013}\\
    \\
    \lineid{013}   &  0003 0030 0210  &   \hide{0012} \hide{0021} 0111 \hide{1011}                      & & &\\
    \lineid{015}   &  0003 0021 0210  &   \hide{0012} \hide{\highlightOne{0021}} \hide{0102} 0111 \hide{0201} \hide{1011} \hide{1101}  &\lineid{014}    & \lineid{013} & \lineid{015}\\
    \lineid{016}   &  0012 0021 0210  &   0012 \hide{0021} \hide{\highlightOne{0030}} 0102 0111 \hide{\highlightOne{0120}} \hide{\highlightOne{0201}} \hide{\highlightOne{0210}} \hide{1011} \hide{\highlightOne{1020}} \hide{\highlightOne{1101}} \hide{\highlightOne{1110}} \hide{\highlightOne{2010}}  &\lineid{006}    & \lineid{016} & \lineid{015}\\
    \\
    \lineid{015}   &  0003 0021 0210  &   \hide{0012} \hide{0021} \hide{0102} 0111 \hide{0201} \hide{1011} \hide{1101}  &    & &\\
    \lineid{017}   &  0012 0021 0210  &   \hide{\highlightOne{0012}} \hide{0021} \hide{0030} \hide{\highlightOne{0102}} 0111 \hide{0120} \hide{0201} \hide{0210} \hide{1011} \hide{1020} \hide{1101} \hide{1110} \hide{2010}  &\lineid{016}    & \lineid{015} & \lineid{017}\\
    \lineid{018}   &  0012 0030 0210  &   0012 \hide{\highlightOne{0021}} 0030 0111 0120 0210 1011 1020 1110 2010  & & \lineid{017} & \lineid{018}\\
    \\
    \lineid{013}   &  0003 0030 0210  &   \hide{0012} \hide{0021} 0111 \hide{1011}                      & & &\\
    \lineid{019}   &  0012 0030 0210  &   0012 \hide{0021} \hide{\highlightOne{0030}} 0111 \hide{\highlightOne{0120}} \hide{\highlightOne{0210}} 1011 \hide{\highlightOne{1020}} \hide{\highlightOne{1110}} \hide{\highlightOne{2010}}  &\lineid{018}    & \lineid{019} & \lineid{013}\\
    \\
    \lineid{013}   &  0003 0030 0210  &   \hide{0012} \hide{0021} 0111 \hide{1011}                      & & &\\
    \lineid{020}   &  0012 0030 0210  &   \hide{\highlightOne{0012}} \hide{0021} \hide{0030} 0111 \hide{0120} \hide{0210} 1011 \hide{1020} \hide{1110} \hide{2010}  &\lineid{019}    & \lineid{013} & \lineid{020}\\
    \\
    \lineid{001}   &  0003 0030 0300  &   0111                                                          & & &\\
    \lineid{021}   &  0003 0030 1110  &   \hide{\highlightOne{0012}} \hide{\highlightOne{0021}} 0111 1011  & & \lineid{021} & \lineid{001}\\
    \\
    \lineid{001}   &  0003 0030 0300  &   0111                                                          & & &\\
    \lineid{022}   &  0012 0030 0300  &   0111 \hide{\highlightOne{0120}} \hide{\highlightOne{0210}} \hide{\highlightOne{1110}}  & & \lineid{022} & \lineid{001}\\
    \lineid{023}   &  0012 0030 0210  &   \hide{0012} \hide{0021} \hide{0030} 0111 \hide{0120} \hide{\highlightOne{1011}} \hide{1020} \hide{0210} \hide{1110} \hide{2010}  &\lineid{020}    & \lineid{023} & \lineid{022}\\
    \lineid{024}   &  0012 0210 3000  &   1011 \hide{\highlightOne{1020}} 1101 1110 2010                & & \lineid{023} & \lineid{024}\\
    \\
    \lineid{001}   &  0003 0030 0300  &   0111                                                          & & &\\
    \lineid{025}   &  0030 0300 1011  &   0111 \hide{\highlightOne{0120}} \hide{\highlightOne{0210}} 1110  & & \lineid{025} & \lineid{001}\\
    \\
    \lineid{026}   &  0003 0030 3000  &   1011                                                          & & &\\
    \lineid{027}   &  0003 0021 3000  &   \hide{\highlightOne{1002}} 1011 \hide{\highlightOne{1101}} \hide{\highlightOne{2001}}  & & \lineid{027} & \lineid{026}\\
    \lineid{028}   &  0003 0021 2010  &   \hide{\highlightOne{0012}} 0021 \hide{\highlightOne{0111}} 1002 1011 1101 2001  & & \lineid{028} & \lineid{027}\\
    \\
    \lineid{027}   &  0003 0021 3000  &   \hide{1002} 1011 \hide{1101} \hide{2001}                     & & &\\
    \lineid{029}   &  0021 1110 3000  &   \hide{\highlightOne{1002}} 1011 1020 1101 1110 \hide{\highlightOne{2001}} 2010  & & \lineid{029} & \lineid{027}\\
    \\
    \lineid{026}   &  0003 0030 3000  &   1011                                                          & & &\\
    \lineid{030}   &  0003 0030 1011  &   \hide{\highlightOne{0012}} \hide{\highlightOne{0021}} \hide{\highlightOne{0111}} 1011  & & \lineid{030} & \lineid{026}\\
    \lineid{031}   &  0030 1011 2010  &   \hide{\highlightOne{0012}} \hide{\highlightOne{0021}} \hide{\highlightOne{0030}} \hide{\highlightOne{0111}} \hide{\highlightOne{0120}} \hide{\highlightOne{0210}} 1011 1020 1110 2010  & & \lineid{031} & \lineid{030}\\
    \\
    \lineid{026}   &  0003 0030 3000  &   1011                                                          & & &\\
    \lineid{032}   &  0003 0030 2010  &   \hide{\highlightOne{0012}} \hide{\highlightOne{0021}} \hide{\highlightOne{0111}} 1011  & & \lineid{032} & \lineid{026}\\
    \lineid{033}   &  0003 0021 2010  &   \hide{0012} 0021 \hide{0111} \hide{\highlightOne{1002}} 1011 \hide{\highlightOne{1101}} \hide{\highlightOne{2001}} &\lineid{028}    & \lineid{033} & \lineid{032}\\
    \\
    \lineid{032}   &  0003 0030 2010  &   \hide{0012} \hide{0021} \hide{0111} 1011                      & & &\\
    \lineid{034}   &  0003 0021 2010  &   \hide{0012} \hide{\highlightOne{0021}} \hide{0111} \hide{1002} 1011 \hide{1101} \hide{2001}  &\lineid{033}    & \lineid{032} & \lineid{034}\\
    \lineid{035}   &  0021 1110 2010  &   \hide{\highlightOne{0012}} \hide{\highlightOne{0021}} \hide{\highlightOne{0030}} 0111 0120 0210 1002 1011 1020 1101 1110 \hide{\highlightOne{2001}} 2010  & & \lineid{035} & \lineid{034}\\
    \\
    \lineid{034}   &  0003 0021 2010  &   \hide{0012} \hide{0021} \hide{0111} \hide{1002} 1011 \hide{1101} \hide{2001}  &\lineid{033}    & &\\
    \lineid{036}   &  0021 1110 2010  &   \hide{0012} \hide{0021} \hide{0030} 0111 0120 0210 \hide{\highlightOne{1002}} 1011 1020 1101 1110 \hide{2001} 2010  &\lineid{035}    & \lineid{034} & \lineid{036}\\
    \\
    \lineid{032}   &  0003 0030 2010  &   \hide{0012} \hide{0021} \hide{0111} 1011                      & & &\\
    \lineid{037}   &  0003 1110 2010  &   \hide{\highlightOne{0012}} \hide{\highlightOne{0021}} 0111 \hide{\highlightOne{1002}} 1011 1101 \hide{\highlightOne{2001}}  & & \lineid{037} & \lineid{032}\\
    \\
    \lineid{032}   &  0003 0030 2010  &   \hide{0012} \hide{0021} \hide{0111} 1011                      & & &\\
    \lineid{038}   &  0030 1002 2010  &   0012 \hide{\highlightOne{0021}} \hide{\highlightOne{0111}} 1011 \hide{\highlightOne{1020}} \hide{\highlightOne{1110}} \hide{\highlightOne{2010}}  & & \lineid{038} & \lineid{032}\\
    \\
    \lineid{032}   &  0003 0030 2010  &   \hide{0012} \hide{0021} \hide{0111} 1011                      & & &\\
    \lineid{039}   &  0030 1002 2010  &   \hide{\highlightOne{0012}} \hide{0021} \hide{0111} 1011 \hide{1020} \hide{1110} \hide{2010}  &\lineid{038}    & \lineid{032} & \lineid{039}\\
    \lineid{040}   &  0030 1002 2100  &   0111 1011 1020 1110 \hide{\highlightOne{2010}}                & & \lineid{039} & \lineid{040}\\
    \\
    \lineid{032}   &  0003 0030 2010  &   \hide{0012} \hide{0021} \hide{0111} 1011                      & & &\\
    \lineid{041}   &  0030 1011 2010  &   \hide{0012} \hide{0021} \hide{0030} \hide{0111} \hide{0120} \hide{0210} 1011 \hide{\highlightOne{1020}} \hide{\highlightOne{1110}} \hide{\highlightOne{2010}}  &\lineid{031}    & \lineid{041} & \lineid{032}\\
    \lineid{042}   &  0030 1011 2100  &   0111 0120 0210 1011 1020 1110 \hide{\highlightOne{2010}}      & & \lineid{041} & \lineid{042}\\
    \\
    \lineid{043}   &  0003 0300 3000  &   1101                                                          & & &\\
    \lineid{044}   &  0003 0201 3000  &   \hide{\highlightOne{1002}} \hide{\highlightOne{1011}} 1101 \hide{\highlightOne{2001}}  & & \lineid{044} & \lineid{043}\\
    \lineid{045}   &  0012 0201 3000  &   \hide{\highlightOne{1002}} 1011 1101 1110 2001                & & \lineid{044} & \lineid{045}\\
    \\
    \lineid{043}   &  0003 0300 3000  &   1101                                                          & & &\\
    \lineid{046}   &  0003 0300 2100  &   \hide{\highlightOne{0102}} \hide{\highlightOne{0111}} \hide{\highlightOne{0201}} 1101  & & \lineid{046} & \lineid{043}\\
    \lineid{047}   &  0003 0210 2100  &   0102 0111 \hide{\highlightOne{0201}} 1011 1101                & & \lineid{046} & \lineid{047}\\
    \\
    \lineid{048}   &  0030 0300 3000  &   1110                                                          & & &\\
    \lineid{049}   &  0030 0300 1110  &   \hide{\highlightOne{0111}} \hide{\highlightOne{0120}} \hide{\highlightOne{0210}} 1110  & & \lineid{049} & \lineid{048}\\
    \lineid{050}   &  0030 0201 1110  &   \hide{0012} \hide{0021} 0111 0120 \hide{\highlightOne{0210}} 1011 1110  &\lineid{012}    & \lineid{049} & \lineid{050}\\
    \\
    \lineid{048}   &  0030 0300 3000  &   1110                                                          & & &\\
    \lineid{051}   &  0030 0300 2010  &   \hide{\highlightOne{0111}} \hide{\highlightOne{0120}} \hide{\highlightOne{0210}} 1110  & & \lineid{051} & \lineid{048}\\
    \lineid{052}   &  0021 0300 2010  &   0111 \hide{\highlightOne{0120}} 1101 1110 0210                & & \lineid{051} & \lineid{052}\\
    \\
    \lineid{048}   &  0030 0300 3000  &   1110                                                          & & &\\
    \lineid{053}   &  0030 0300 2100  &   \hide{\highlightOne{0111}} \hide{\highlightOne{0120}} \hide{\highlightOne{0210}} 1110  & & \lineid{053} & \lineid{048}\\
    \lineid{054}   &  0030 0201 2100  &   0111 0120 \hide{\highlightOne{0210}} 1011 1110                & & \lineid{053} & \lineid{054}\\
    \\
    \lineid{053}   &  0030 0300 2100  &   \hide{0111} \hide{0120} \hide{0210} 1110                      & & &\\
    \lineid{055}   &  0030 1011 2100  &   0111 \hide{\highlightOne{0120}} 0210 1011 1020 1110 \hide{2010}  &\lineid{042}    & \lineid{055} & \lineid{053}\\
    \\
    \lineid{053}   &  0030 0300 2100  &   \hide{0111} \hide{0120} \hide{0210} 1110                     & & &\\
    \lineid{056}   &  0030 1011 2100  &   0111 \hide{0120} \hide{\highlightOne{0210}} 1011 1020 1110 \hide{2010}  &\lineid{055}    & \lineid{053} & \lineid{056}\\
    \\
        \hline
\end{longtable}

According to \lineid{054}, the profile 0030 0201 2100 has four remaining possible output allocations: 0111, 1011, 1110, and 0120. We make a case distinction.\\

\textbf{Case 1:} 0030 0201 2100 has outcome 0111.
\begin{longtable}{@{\hspace{-2mm}}>{\centering\arraybackslash}p{1.2cm} l p{8cm} >{\centering\arraybackslash}m{1.6cm} >{\centering\arraybackslash}m{1.2cm} >{\centering\arraybackslash}m{1.2cm}@{\hspace{2mm}}}
    \hline
    ID & Profile & Outcomes & Last Mentioned & Truthful ID & Misreport ID \\
    \hline
    \\
    \lineid{054}   &  0030 0201 2100  &   0111 \hide{\highlightOne{0120}} \hide{0210} \hide{\highlightOne{1011}} \hide{\highlightOne{1110}}          & & &\\
    \lineid{057}   &  0030 0201 1011  &   \hide{0012} 0021 0111 \hide{0120} \hide{0210} 1011 \hide{\highlightOne{1110}}  &\lineid{011}    & \lineid{054} & \lineid{057}\\
    \\
    \lineid{054}   &  0030 0201 2100  &   0111 \hide{0120} \hide{0210} \hide{1011} \hide{1110}          & & &\\
    \lineid{058}   &  0030 0201 3000  &   1011 \hide{\highlightOne{1110}}                               & & \lineid{054} & \lineid{058}\\
    \lineid{059}   &  0012 0201 3000  &   \hide{1002} 1011 \hide{\highlightOne{1101}} \hide{\highlightOne{1110}} \hide{\highlightOne{2001}}  &\lineid{045}    & \lineid{059} & \lineid{058}\\
    \lineid{060}   &  0012 0210 3000  &   1011 \hide{1020} \hide{\highlightOne{1101}} 1110 2010         &\lineid{024}    & \lineid{059} & \lineid{060}\\
    \\
    \lineid{059}   &  0012 0201 3000  &   \hide{1002} 1011 \hide{1101} \hide{1110} \hide{2001}          &    & &\\
    \lineid{061}   &  0012 0300 3000  &   \hide{\highlightOne{1101}} 1110                               & & \lineid{059} & \lineid{061}\\
    \lineid{062}   &  0012 0210 3000  &   \hide{\highlightOne{1011}} \hide{1020} \hide{1101} 1110 \hide{\highlightOne{2010}}  &\lineid{060}    & \lineid{062} & \lineid{061}\\
    \lineid{063}   &  0003 0210 3000  &   \hide{\highlightOne{1011}} 1101                               & & \lineid{062} & \lineid{063}\\
    \lineid{064}   &  0003 0210 2100  &   \hide{\highlightOne{0102}} \hide{\highlightOne{0111}} \hide{0201} \hide{\highlightOne{1011}} 1101  &\lineid{047}    & \lineid{064} & \lineid{063}\\
    \lineid{065}   &  0003 0030 2100  &   \hide{\highlightOne{0111}} 1011                               & & \lineid{064} & \lineid{065}\\
    \lineid{066}   &  0030 1011 2100  &   0111 \hide{0120} \hide{0210} 1011 1020 \hide{\highlightOne{1110}} \hide{2010}  &\lineid{056}    & \lineid{066} & \lineid{065}\\
    \\
    \lineid{058}   &  0030 0201 3000  &   1011 \hide{1110}                                              & & &\\
    \lineid{067}   &  0030 0201 1011  &   \hide{0012} \hide{\highlightOne{0021}} \hide{\highlightOne{0111}} \hide{0120} \hide{0210} 1011 \hide{1110}  &\lineid{057}    & \lineid{067} & \lineid{058}\\
    \lineid{068}   &  0030 0300 1011  &   \hide{\highlightOne{0111}} \hide{0120} \hide{0210} 1110      &\lineid{025}    & \lineid{067} & \lineid{068}\\
    \lineid{069}   &  0030 1011 2100  &   0111 \hide{0120} \hide{0210} \hide{\highlightOne{1011}} \hide{\highlightOne{1020}} \hide{1110} \hide{2010}  &\lineid{066}    & \lineid{069} & \lineid{068}\\
    \lineid{070}   &  0030 0201 1011  &   \hide{0012} \hide{0021} \hide{0111} \hide{0120} \hide{0210} \hide{\highlightOne{1011}} \hide{1110}  &\lineid{067}    & \lineid{070} & \lineid{069}\\
            \\
    
        \hline
\end{longtable}
\textbf{Case 2:} 0030 0201 2100 has outcome 1011.
\begin{longtable}{@{\hspace{-2mm}}>{\centering\arraybackslash}p{1.2cm} l p{8cm} >{\centering\arraybackslash}m{1.6cm} >{\centering\arraybackslash}m{1.2cm} >{\centering\arraybackslash}m{1.2cm}@{\hspace{2mm}}}
    \hline
    ID & Profile & Outcomes & Last Mentioned & Truthful ID & Misreport ID \\
    \hline
    \\
    \lineid{054}   &  0030 0201 2100  &   \hide{\highlightOne{0111}} \hide{\highlightOne{0120}} \hide{0210} 1011 \hide{\highlightOne{1110}}          & & &\\
    \lineid{071}   &  0030 0201 1011  &   \hide{0012} \hide{\highlightOne{0021}} \hide{\highlightOne{0111}} \hide{0120} \hide{0210} 1011 \hide{\highlightOne{1110}}  &\lineid{011}    & \lineid{071} & \lineid{054}\\
    \lineid{072}   &  0030 0300 1011  &   \hide{\highlightOne{0111}} \hide{0120} \hide{0210} 1110      &\lineid{025}    & \lineid{071} & \lineid{072}\\
    \lineid{073}   &  0030 1011 2100  &   0111 \hide{0120} \hide{0210} \hide{\highlightOne{1011}} \hide{\highlightOne{1020}} 1110 \hide{2010}  &\lineid{056}    & \lineid{073} & \lineid{072}\\
    \\
    \lineid{071}   &  0030 0201 1011  &   \hide{0012} \hide{0021} \hide{0111} \hide{0120} \hide{0210} 1011 \hide{1110}  &    & &\\
    \lineid{074}   &  0030 1011 2100  &   \hide{\highlightOne{0111}} \hide{0120} \hide{0210} \hide{1011} \hide{1020} 1110 \hide{2010}  &\lineid{073}    & \lineid{071} & \lineid{074}\\
    \lineid{075}   &  0030 0201 2100  &   \hide{0111} \hide{0120} \hide{0210} \hide{\highlightOne{1011}} \hide{1110}  &\lineid{054}    & \lineid{074} & \lineid{075}\\
        \\
    
        \hline
\end{longtable}
\textbf{Case 3:} 0030 0201 2100 has outcome 1110.
\begin{longtable}{@{\hspace{-2mm}}>{\centering\arraybackslash}p{1.2cm} l p{8cm} >{\centering\arraybackslash}m{1.6cm} >{\centering\arraybackslash}m{1.2cm} >{\centering\arraybackslash}m{1.2cm}@{\hspace{2mm}}}
    \hline
    ID & Profile & Outcomes & Last Mentioned & Truthful ID & Misreport ID \\
    \hline
    \\
    \lineid{054}   &  0030 0201 2100  &   \hide{\highlightOne{0111}} \hide{\highlightOne{0120}} \hide{0210} \hide{\highlightOne{1011}} 1110          & & &\\
    \lineid{076}   &  0003 0030 2100  &   \hide{\highlightOne{0111}} 1011                               & & \lineid{054} & \lineid{076}\\
    \lineid{077}   &  0030 1002 2100  &   0111 1011 \hide{\highlightOne{1020}} \hide{\highlightOne{1110}} \hide{2010}  &\lineid{040}    & \lineid{077} & \lineid{076}\\
    \\
    \lineid{054}   &  0030 0201 2100  &   \hide{0111} \hide{0120} \hide{0210} \hide{1011} 1110          & & &\\
    \lineid{078}   &  0030 0201 1110  &   \hide{0012} \hide{0021} \hide{\highlightOne{0111}} \hide{\highlightOne{0120}} \hide{0210} \hide{\highlightOne{1011}} 1110  &\lineid{050}    & \lineid{078} & \lineid{054}\\
    \lineid{079}   &  0003 0030 1110  &   \hide{0012} \hide{0021} \hide{\highlightOne{0111}} 1011       &\lineid{021}    & \lineid{078} & \lineid{079}\\
    \lineid{080}   &  0003 1110 2010  &   \hide{0012} \hide{0021} \hide{\highlightOne{0111}} \hide{1002} 1011 \hide{\highlightOne{1101}} \hide{2001}  &\lineid{037}    & \lineid{080} & \lineid{079}\\
    \lineid{081}   &  0021 1110 2010  &   \hide{0012} \hide{0021} \hide{0030} 0111 0120 \hide{\highlightOne{0210}} \hide{1002} 1011 1020 \hide{\highlightOne{1101}} \hide{\highlightOne{1110}} \hide{2001} \hide{\highlightOne{2010}}  &\lineid{036}    & \lineid{081} & \lineid{080}\\
    \\
    \lineid{054}   &  0030 0201 2100  &   \hide{0111} \hide{0120} \hide{0210} \hide{1011} 1110          & & &\\
    \lineid{082}   &  0030 1002 2100  &   \hide{\highlightOne{0111}} 1011 \hide{1020} \hide{1110} \hide{2010}  &\lineid{077}    & \lineid{054} & \lineid{082}\\
    \lineid{083}   &  0030 0300 1002  &   0111 \hide{\highlightOne{1110}}                               & & \lineid{082} & \lineid{083}\\
    \lineid{084}   &  0021 0300 1002  &   \hide{0102} 0111 \hide{\highlightOne{0201}} \hide{\highlightOne{1101}} \hide{\highlightOne{1110}}  &\lineid{007}    & \lineid{084} & \lineid{083}\\
    \lineid{085}   &  0021 0300 2010  &   0111 \hide{0120} 0210 \hide{\highlightOne{1101}} 1110        &\lineid{052}    & \lineid{084} & \lineid{085}\\
    \\
    \lineid{084}   &  0021 0300 1002  &   \hide{0102} 0111 \hide{0201} \hide{1101} \hide{1110}          &    & &\\
    \lineid{086}   &  0021 0300 3000  &   \hide{\highlightOne{1101}} 1110                               & & \lineid{084} & \lineid{086}\\
    \lineid{087}   &  0021 0300 2010  &   \hide{\highlightOne{0111}} \hide{0120} \hide{\highlightOne{0210}} \hide{1101} 1110  &\lineid{085}    & \lineid{087} & \lineid{086}\\
    \lineid{088}   &  0021 1110 2010  &   \hide{0012} \hide{0021} \hide{0030} 0111 0120 \hide{0210} \hide{1002} \hide{\highlightOne{1011}} \hide{\highlightOne{1020}} \hide{1101} \hide{1110} \hide{2001} \hide{2010}  &\lineid{081}    & \lineid{088} & \lineid{087}\\
    \\
    \lineid{093}   &  0021 0300 3000  &   \hide{1101} 1110                                              & & &\\
    \lineid{089}   &  0021 1110 3000  &   \hide{1002} \hide{\highlightOne{1011}} \hide{\highlightOne{1020}} \hide{\highlightOne{1101}} 1110 \hide{2001} \hide{\highlightOne{2010}}  &\lineid{029}    & \lineid{089} & \lineid{093}\\
    \lineid{090}   &  0021 1110 2010  &   \hide{0012} \hide{0021} \hide{0030} \hide{\highlightOne{0111}} \hide{\highlightOne{0120}} \hide{0210} \hide{1002} \hide{1011} \hide{1020} \hide{1101} \hide{1110} \hide{2001} \hide{2010}  &\lineid{088}    & \lineid{090} & \lineid{089}\\
    \\    
        \hline
\end{longtable}
\textbf{Case 4:} 0030 0201 2100 has outcome 0120.
\begin{longtable}{@{\hspace{-2mm}}>{\centering\arraybackslash}p{1cm} l p{8cm} >{\centering\arraybackslash}m{1.6cm} >{\centering\arraybackslash}m{1.2cm} >{\centering\arraybackslash}m{1.2cm}@{\hspace{2mm}}}
    \hline
    ID & Profile & Outcomes & Last Mentioned & Truthful ID & Misreport ID \\
    \hline
    \\
    \lineid{054}   &  0030 0201 2100  &   \hide{\highlightOne{0111}} 0120 \hide{0210} \hide{\highlightOne{1011}} \hide{\highlightOne{1110}}          & & &\\
    \lineid{091}   &  0030 0201 1002  &   \hide{0012} 0021 0111 1011 \hide{\highlightOne{1110}}         &\lineid{009}    & \lineid{054} & \lineid{091}\\
    \\
    \lineid{054}   &  0030 0201 2100  &   \hide{0111} 0120 \hide{0210} \hide{1011} \hide{1110}          & & &\\
    \lineid{092}   &  0030 0201 3000  &   1011 \hide{\highlightOne{1110}}                               & & \lineid{054} & \lineid{092}\\
    \lineid{093}   &  0012 0201 3000  &   \hide{1002} 1011 \hide{\highlightOne{1101}} \hide{\highlightOne{1110}} \hide{\highlightOne{2001}}  &\lineid{045}    & \lineid{093} & \lineid{092}\\
    \lineid{094}   &  0012 0210 3000  &   1011 \hide{1020} \hide{\highlightOne{1101}} 1110 2010         &\lineid{024}    & \lineid{093} & \lineid{094}\\
    \\
    \lineid{093}   &  0012 0201 3000  &   \hide{1002} 1011 \hide{1101} \hide{1110} \hide{2001}          &    & &\\
    \lineid{095}   &  0012 0300 3000  &   \hide{\highlightOne{1101}} 1110                               & & \lineid{093} & \lineid{095}\\
    \lineid{096}   &  0012 0210 3000  &   \hide{\highlightOne{1011}} \hide{1020} \hide{1101} 1110 \hide{\highlightOne{2010}}  &\lineid{094}    & \lineid{096} & \lineid{095}\\
    \lineid{097}   &  0003 0210 3000  &   \hide{\highlightOne{1011}} 1101                               & & \lineid{096} & \lineid{097}\\
    \lineid{098}   &  0003 0210 2100  &   \hide{\highlightOne{0102}} \hide{\highlightOne{0111}} \hide{0201} \hide{\highlightOne{1011}} 1101  &\lineid{047}    & \lineid{098} & \lineid{097}\\
    \lineid{099}   &  0003 0030 2100  &   \hide{\highlightOne{0111}} 1011                               & & \lineid{098} & \lineid{099}\\
    \lineid{100}   &  0030 1002 2100  &   0111 1011 \hide{\highlightOne{1020}} \hide{\highlightOne{1110}} \hide{2010}  &\lineid{040}    & \lineid{100} & \lineid{099}\\
    \\
    \lineid{092}   &  0030 0201 3000  &   1011 \hide{1110}                                              & & &\\
    \lineid{101}   &  0030 0201 1002  &   \hide{0012} \hide{\highlightOne{0021}} \hide{\highlightOne{0111}} 1011 \hide{1110}  &\lineid{091}    & \lineid{101} & \lineid{092}\\
    \lineid{102}   &  0030 0300 1002  &   \hide{\highlightOne{0111}} 1110                               & & \lineid{101} & \lineid{102}\\
    \lineid{103}   &  0030 1002 2100  &   0111 \hide{\highlightOne{1011}} \hide{1020} \hide{1110} \hide{2010}  &\lineid{100}    & \lineid{103} & \lineid{102}\\
    \lineid{104}   &  0030 0201 2100  &   \hide{0111} \hide{\highlightOne{0120}} \hide{1011} \hide{1110} \hide{0210}  &\lineid{054}    & \lineid{104} & \lineid{103}\\
        \\
    
    \hline
\end{longtable}

    We have reached a contradiction in all four cases, which completes the proof.
\end{proof}

\section{Missing Definitions and Proofs from \Cref{sec:continuous_input}} \label{app:continuous_input}

We start this section by arguing that a strengthening of \Cref{thm:strict_order_impossibility} holds in \Cref{sec:strengthening}, before using this result for our proof of \Cref{thm:continuous_impossibility} which we present in \Cref{app-sec:continupus-impossibility}.

\subsection{Strengthening the Result by \citet{ACS03b}} \label{sec:strengthening}

Let $A$ be a set of alternatives and $\mathcal{R}(A)$ be the set of all weak rankings over $A$. 
For any weak ranking $\wpo \in \mathcal{R}(A)$, we let $r_1(\wpo) = \{a \in A \mid a \wpo a' \text{ for all } a' \in A\}$ and $r_2(\wpo) = \{a \in A \setminus r_1(\wpo)  \mid a \wpo a' \text{ for all } a' \in A \setminus r_1(\wpo)\}$ be the set of alternatives of \textit{rank} $1$ and $2$, respectively, with respect to $\wpo$. 
We call $\mathbb{D} \subseteq \mathcal{R}(A)$ a (sub)domain. 
In particular, we will be interested in the subdomain of weak rankings for which the first and second ranks are unique, i.e., $\hat{\mathcal{R}}(A) = \{\wpo \in \mathcal{R}(A) \mid |r_1(\wpo)| = |r_2(\wpo)| = 1\}$. 
For $\hat{\mathcal{R}}(A)$, \textit{linkedness} can be defined in an analogous manner as for strict rankings in \Cref{sec:continuous_input}. 
For completeness, we restate this definition below. 

\paragraph{Linked Domains.}
Let $\mathbb{D} \subseteq \hat{\mathcal{R}}(A)$ be a subdomain. 
\begin{itemize}
    \item We call two alternatives $a,a' \in A$ \emph{connected} in $\mathbb{D}$ if there exist weak rankings $\wpo,\wpo' \in \mathbb{D}$ such that $a \in r_1(\wpo) =  r_2(\wpo')$ and $a' \in r_1(\wpo') =  r_2(\wpo)$.
    \item We say that alternative $a \in A$ is \emph{linked} to a subset $B \subseteq A$ if there exist distinct $a', a'' \in B$ such that $a$ is connected to both $a'$ and $a''$ in $\mathbb{D}$.
    \item We call the subdomain $\mathbb{D}$ \emph{linked} if we can order the alternatives in $A$ into a vector $(a^{1}, \dots, a^{|A|})$, such that
    $a^{1}$ is connected to $a^{2}$ and, for all $k \in \{3, \dots, |A|\}$,
    it holds that $a^{k}$ is linked to $\{a^{1}, \dots, a^{k-1}\}$.
\end{itemize}

Before we formalize the strengthening of the result by \citet{ACS03b}, we define social choice functions along with the three required axioms. 

\paragraph{Social Choice Function.} For a domain $\mathbb{D} \subseteq \mathcal{R}(A)$, a social choice function is a family of functions\footnote{Note, that \citet{ACS03b} define a social choice function with a fixed number of voters $n$, whereas we define a social choice function as a family of functions, one for each number of voters.} $\B_n: \mathbb{D}^n \rightarrow A$, one for every number of voters $n \in \N$.
Since $n$ is often clear from context, we slightly abuse notation and write $\B$ instead of $\B_n$. 

\paragraph{Unanimity.} A social choice function $\B$ is \emph{unanimous} if for any number of voters $n \in \N$, alternative $a \in A$ and profile $P = (\wpo_1, \dots, \wpo_n)$ with $r_1(\wpo_i) = \{a\}$ for all $i \in [n]$ we have $\B(P) = a$.

\paragraph{Dictatorial.} For $n \in \N$, voter $i \in [n]$ is a \textit{dictator} for a social choice function $\B$ if for all profiles $P = (\wpo_1, \dots, \wpo_n)$ with $|r_1(\wpo_i)| = 1$ we have $\B(P) \in r_1(\wpo_i)$. A social choice function $\B$ is called \textit{dictatorial} for $n \in \N$ if, for any $n$ voters, there exists a voter that is a dictator for $\B$.

\paragraph{Truthfulness.} A social choice function $\B$ is \emph{truthful} if for any number of voters $n \in \N$, any profile $P = (\wpo_1, \dots, \wpo_n)$, voter $i \in [n]$, and misreport $\wpo_i^\star$, the following holds for profile $P^\star = (\wpo_1, \dots, \wpo_{i-1}, \wpo_i^\star, \wpo_{i+1}, \dots, \wpo_n)$: \[\B(P) \wpo_i \B(P^\star).\]

\begin{theorem}[{Adjusted version of Theorem 3.1 of \citet{ACS03b}}] \label{thm:strict_order_impossibility-strong}
    For any set of alternatives $A$ with $|A|\ge 3$, the following holds: If a subdomain $\mathbb{D} \subseteq \hat{\mathcal{R}}(A)$ is linked, then any unanimous and truthful social choice function $\B$ on domain $\mathbb{D}$ is dictatorial for any number of voters $n \in \N$.
\end{theorem}

\begin{proof}[Proof Sketch]
    We claim that the proof provided by \citet{ACS03b} carries over essentially as it is written from the class of strict rankings $\mathcal{L}(A)$ to the class $\hat{\mathcal{R}}(A)$, i.e., the class of weak rankings $\wpo$ with the property that $r_1(\wpo)$ and $r_2(\wpo)$ are singletons. Since the proof by \citet{ACS03b} spans five pages in their paper, it would be out of scope to identically reproduce this proof at this point. To nevertheless be more concrete, we roughly summarize the proof below and point out the situations at which properties of strict rankings are used by \citet{ACS03b} and why they are also met by $\hat{\mathcal{R}}(A)$. 

    \paragraph{Proof of Proposition 3.1 of \citet{ACS03b}.} The proof by \citet{ACS03b} starts by showing that the statement in  \Cref{thm:strict_order_impossibility} holds if and only if the same statement holds for the case of two voters. This statement is proven in Proposition 3.1 of their paper. We discuss two points in this proof: 
    \begin{itemize}
        \item Since the forward direction is trivial, they focus on the backwards direction. To this end, they first assume that $f$ is a truthful social choice function satisfying unanimity. From that, they define a social choice function $g$ for two voters by defining $g(\triangleright_i,\triangleright_j) = f(\triangleright_i,\triangleright_j, \dots, \triangleright_j)$. They then go on and argue that $g$ is truthful as well. For contradiction, they assume that $g$ is manipulable by $j$, i.e., there exist $\triangleright_i, \triangleright_j,\triangleright_j'$ such that $b = g(\triangleright_i, \triangleright_j') \triangleright_j g(\triangleright_i, \triangleright_j) = a$. Then, they argue that sequentially moving from $(\triangleright_i, \triangleright_j, \dots, \triangleright_j)$ to $(\triangleright_i, \triangleright_j', \dots, \triangleright_j')$ and applying the truthfulness of $f$ in each step shows that $a = f(\triangleright_i, \triangleright_j, \dots, \triangleright_j) \triangleright_j f(\triangleright_i, \triangleright_j', \dots, \triangleright_j') = b$, which is a contradiction to the previous assumption. For the case of weak rankings, we can still assume that $b \triangleright_j a$, but only show that $a \wpo_j b$. Clearly, this still yields a contradiction.
        \item On the second page of this proof, the authors use the fact that a social choice function that is onto and truthful also satisfies unanimity. This implication also holds for social choice functions defined on $\hat{\mathcal{R}}(A)$. 
    \end{itemize}

    \paragraph{Induction over the alternatives.} After having established Proposition 3.1, the proof moves on and shows \Cref{thm:strict_order_impossibility} for the case of two voters. The proof carries out an induction over the number of alternatives. Within this induction, the proof uses the concept of option sets. The set $O_2(\triangleright_1) \subseteq A$ contains all alternatives that can be returned, given that voter $1$ votes $\triangleright_1$, i.e., it is the \emph{option set} for voter $2$. The proof now uses the fact that any truthful social choice function has to satisfy $f(\triangleright_1,\triangleright_2) = \max_{\triangleright_2}(O_2(\triangleright_1))$ (and the same for reversed roles of the voters). This is a well-established fact which can be easily shown (for example, see \citet{barbera1990strategy}). For the case of weak rankings, this generalizes to $f(\wpo_1,\wpo_2) \in \max_{\wpo_2}(O_2(\wpo_1))$. However, since the proof only applies this result when at least one of the two top-ranked alternatives of voter $2$ is in their option set (and thus, the maximum set is a singleton), we can use the original version of the statement in all of these cases. 
\end{proof}

\subsection{Proof of \Cref{thm:continuous_impossibility}} \label{app-sec:continupus-impossibility}

We now discuss the connection between fractional-input mechanisms and weak rankings. For that, let $n,m,b$ be fixed for the rest of the section. For each vote $p \in \C$ let $\wpo_p$ be the inferred weak ranking over all integral allocations in $\D$, i.e., for two integral allocations $a, a' \in \D$ we have $a \wpo_p a'$ if and only if $\ellone{p}{a} \le \ellone{p}{a'}$. Let $\nabla = \{\wpo_p \mid p \in \C \text{ with } |r_1(\wpo_p)| = |r_2(\wpo_p)| = 1\}$ be the domain of weak rankings over the elements of $\D$ inferred from all elements in $\C$ with unique rank one and two allocations. We show in the following that $\nabla$ forms a linked domain. 

\begin{lemma}\label{lem:nabla-linked}
    The subdomain $\nabla \subseteq \hat{\mathcal{R}}(A)$ is linked. 
\end{lemma}

\begin{proof}
We prove the claim by iteratively constructing a vector $(a^{1}, \dots, a^{|\D|})$ of the allocations in $\D$. For each allocation added, we argue why it is connected in $\nabla$ to at least two previously added allocations. 

 We call two integral allocations $a,a' \in \D$ \textit{adjacent} if they have $\ell_1$-distance of 2, which means they only differ on two alternatives and only by $1$ each. Note that for any pair of adjacent allocations $a,a' \in \D$, we can find a fractional vote $p \in \C$, such that $r_1(\wpo_p) = \{a\}$ and $r_2(\wpo_p) = \{a'\}$.\footnote{Choose $p = \frac{2}{3} a + \frac{1}{3} a'$, which gives $\ellone{p}{a} = \frac{2}{3}$, $\ellone{p}{a'} = \frac{4}{3}$, and $\ellone{p}{\hat{a}} \ge 2$ for any other integral allocation $\hat{a} \in \D$.} Thus, all pairs of adjacent allocations are connected in $\nabla$. By adding allocations from $\D$ to the vector, such that they are adjacent to at least two previous allocations, we make sure to satisfy the connection requirements from the definition of a linked domain.

We add the allocations to the vector in three phases. We provide an example of the vector construction for $m = 4$ and $b = 3$ in \Cref{tab:connected_domain_order}. 

        \begin{table}[h]
        \centering
        \caption{Example construction of the ranking of the elements of $I^4_3$. 
        For readability, we shorten the notation for an allocation by leaving out parenthesis, commas and spaces, e.g., the allocation $(1,2,0,0)$ is written as $1200$. In the first phase, we add allocations $a$ with $a_3 \in \{0,1\}$ and $a_4 = 0$. In the second phase, we add allocations $a$ with $a_4 = 0$ in increasing order of $a_3$. In the third phase, we add allocations $a$ with $a_4 \neq 0$ in increasing order of $a_4$.}
        \label{tab:connected_domain_order}
        \begin{tabular}{l|l}
            \textbf{Phase 1} & 3000, 2010, 2100, 1110, 1200, 0210, 0300\\
            \textbf{Phase 2} & 1020, 0120\\
                             & 0030\\
            \textbf{Phase 3} & 2001, 1101, 1011, 0201, 0111, 0021\\
                             & 1002, 0102, 0012\\
                             & 0003\\
        \end{tabular}
    \end{table}
    
    In the \textit{first phase} we only add allocations $a$ with $a_3 \in \{0, 1\}$ and $a_4 = \dots = a_m = 0$. The first two allocations are $(b, 0, 0, \dots, 0)$ and $(b-1, 0, 1, 0, \dots, 0)$. We then alternate, between (i) adding $1$ to the second element and subtracting $1$ from the third and (ii) adding $1$ to the third element and subtracting $1$ from the first alternative, until we added $(0, b, 0, \dots, 0)$. In both cases, we can easily see, that the new element is adjacent to the two previous elements.

    For the \textit{second phase}, we add all remaining allocations $a$ with $a_4 = \dots = a_m = 0$ in increasing order of $a_3 = 2, \dots, b$. For any $a = (a_1, a_2, a_3, 0, \dots, 0)$, we know that the two adjacent allocations $\hat{a} = (a_1+1, a_2, a_3-1, 0, \dots, 0)$ and $\bar{a} = (a_1, a_2+1, a_3-1, 0, \dots, 0)$ have already been added.

    Finally, in the \textit{third phase}, for each remaining alternative $j \in \{4, \dots, m\}$, we add all elements with $a_{j+1} = \dots = a_m = 0$ in increasing order of $a_j = 0, \dots, b$. As in the second phase, we know for any $a = (a_1, a_2, \dots, a_{j-1}, a_j, 0, \dots, 0)$ that the adjacent allocations $\hat{a} = (a_1+1, a_2, \dots, a_{j-1}, a_j-1, 0, \dots, 0)$ and $\bar{a} = (a_1, a_2+1, \dots, a_{j-1}, a_j-1, 0, \dots, 0)$ have already been added.

We constructed a vector of all the allocations in $\D$, such that each one is connected in $\nabla$ to two previous allocations, which shows that $\nabla$ is linked. \end{proof}

Now, we restate \Cref{thm:continuous_impossibility} from \Cref{sec:continuous_input}. Recall that the axioms \textit{onto} and \textit{dictatorial} were defined in the main text. 

\continuousImpossibility*

\begin{proof}

The high-level structure of the proof is the following: 
\begin{enumerate}[label=(\roman*)]
    \item{We assume that there exists a fractional-input mechanism $\mathcal{A}$ that is onto, truthful, and non-dictatorial for some $n \in \N$. Then, we show that such a mechanism $\mathcal{A}$ is in particular also unanimous (definition follows below).} \label{step1}
    \item We show that the fractional-input mechanism $\mathcal{A}$ induces a social choice function $\mathcal{B}$ on the domain $\nabla$ that is unanimous, truthful, and non-dictatorial for $n$ voters. However, since $\nabla$ is linked by \Cref{lem:nabla-linked}, this yields a contradiction to the assumption from \ref{step1}.
\end{enumerate}

\medskip

\paragraph{Step 1.} We assume that there exists a fractional-input mechanism $\mathcal{A}$ that is onto, truthful, and non-dictatorial  for some $n \in \N$. We start by defining unanimity.

 \paragraph{Unanimity.} For an allocation $a \in \D$, let $X_a \subset \Cc$ be the set of votes that strictly prefer $a$ over any other allocation in $\D$. A fractional-input mechanism $\A$ is \emph{unanimous} if for any $n, m, b \in \N$ with $m \ge 2$, allocation $a \in \D$ and profile $P \in (X_a)^{n}$ we have $\A(P) = a$.

\medskip 

We now show that $\A$ being onto and truthful implies that $\A$ is unanimous. We remark that this implication has been established in other contexts. For $a \in \D$, let $X_a \subseteq \C$ be as in the definition of unanimity. Assume for contradiction that there exists some $a \in \D$ and a profile $P=(p_1, \dots, p_n) \in (X_a)^n$ with $\A(P) = a' \neq a$. Since $\A$ is onto there must be another profile $\hat{P} = (\hat{p}_1, \dots, \hat{p}_n)$ with $\A(\hat{P}) = a$. We transform the profile $P$ into the profile $\hat{P}$, by moving each voter $i$ from their original vote $p_i$ to $\hat{p}_i$ one by one. For $i \in [n]_0$ let $\Bar{P}_i = (\hat{p}_1, \dots, \hat{p}_i, p_{i+1}, \dots, p_n)$ be the profile in which every voter up to $i$ has changed their vote from $p_i$ to $\hat{p}_i$. Note that $\Bar{P}_0 = P$ and $\Bar{P}_n = \hat{P}$. We know by truthfulness for any $i \in [n-1]_0$ that if $\A(\Bar{P}_i) \neq a$ then $\A(\Bar{P}_{i+1}) \neq a$, otherwise voter $(i+1)$ with a truthful vote of $p_{i+1} \in X_a$ would prefer the outcome of the profile with a misreport $a = \A(\Bar{P}_{i+1})$ over the outcome of the profile with the truthful report $\A(\Bar{P}_i) \neq a$. By iteratively applying this argument, we get that $\A(\Bar{P}_i) \neq a$ for all $i \in [n]_0$, contradicting the assumption $a = \A(\hat{P}) = A(\Bar{P}_n)$. Hence, $\A$ is unanimous. 

\bigskip

    \paragraph{Step 2.} We first define for every element in $\wpo \in \nabla$ exactly one representative in $\C$. More precisely, let $\pi: \nabla \rightarrow \C$ be such that for $p = \pi(\wpo)$ it holds that $\wpo_{p} = \wpo$. Then, we define the set of representatives by $\Cn = \{\pi(\wpo) \mid \wpo \in \nabla\}$. 
    Now, we are ready to construct the social choice function $\mathcal{B}$ on domain $\nabla$ as follows: for any profile of weak rankings $P = (\wpo_1. \dots, \wpo_n) \in \nabla^n$, we return the output $\B(P) = \A(\pi(\wpo_1), \dots, \pi(\wpo_n))$.

    In the following, we are going to show that $\mathcal{B}$ is unanimous, truthful, and non-dictatorial for $n$ voters by using these properties of $\A$. We point out that we defined these properties independently for the two functions, since the former is a social choice function and the latter is a fractional-input mechanism. 
    
    \paragraph{Claim:} $\B$ is unanimous, truthful, and non-dictatorial for $n$ voters. 

    \medskip

    \noindent 
    \textit{Proof of claim:} The fact that $\A$ is unanimous directly implies that $\B$ is unanimous. The same holds for truthfulness. 

    It remains to show that $\B$ is non-dictatorial for $n$. Suppose without loss of generality that voter $n$ is a dictator under $\B$. Since $\A$ is non-dictatorial for $n$ voters, there must be an integral allocation $a \in \D$ and a profile $P = (p_1, \dots, p_n) \in \Cc$ with an aggregate $a^\star = \A(P)$, such that $\ellone{p_n}{a} < \ellone{p_n}{a^\star}$ (i.e., a certificate of $n$ not being a dictator). Let $p^\star \in \Cn \cap X_{a^{\star}}$, i.e., $p^{\star}$ is a vote that ``prefers'' $a^\star$ over any other allocation from $\D$ and is a representative. We iteratively transform the profile $P$ into the profile $P^\star = (p^\star, \dots, p^\star, p_n)$ by moving each voter $i \in [n-1]$ from $p_i$ to $p^\star$. In each step, we argue that the aggregate cannot change. This is because otherwise the $i$th voter with (truthful) vote $p^\star$ in the new profile $(p^\star, \dots, p^\star, p_{i+1}, \dots, p_n)$ with an aggregate $a' \neq a^\star$ could misreport $p_i$ to attain the profile $(p^\star, \dots, p^\star, p_i, \dots, p_n)$ with aggregate $a^\star$, which is strictly preferred by voter $i$, contradicting truthfulness of $\A$. Thus, we have $\A(P^\star) = a^\star$. Consider the situation when voter $n$ misreports any element $p_n' \in \Cn \cap X_a$. The new profile $P' = (p^\star, \dots, p^\star, p_n')$ is $(\Cn)^n$. Therefore, by the fact that $n$ is a dictator under $\B$ and $p_n' \in X_a$, we know that $\B(\wpo_{p^\star}, \wpo_{p^\star}, \dots, \wpo_{p_n'}) = a$. By the construction of $\B$ this also implies that $\A(P') = a$. However, this contradicts truthfulness of $\A$. Thus, $\B$ is non-dictatorial for $n$ voters. 
    \hfill $\blacksquare$

\medskip

As sketched out before, we showed in Step 2 that $\B$ is unanimous, truthful, and non-dictatorial for $n$ voters. However, since $\nabla$ is a linked domain (by \Cref{lem:nabla-linked}) over $|\D| \ge 3$ alternatives (since $(m-1) \cdot b \ge 2)$, this is a contradiction to \Cref{thm:strict_order_impossibility-strong}. Therefore, $\A$ has to be dictatorial for all $n \in \N$. 
\end{proof}

\end{document}